\documentclass[a4paper,onecolumn,10pt,accepted=2021-05-01]{quantumarticle}
\pdfoutput=1
\usepackage[utf8]{inputenc}
\usepackage[english]{babel}
\usepackage[T1]{fontenc}
\usepackage{amsmath}
\usepackage{hyperref}
\usepackage{tikz}
\usepackage{lipsum}
\usepackage{times,amsmath,amsfonts,amssymb,latexsym,amsthm}
\usepackage{graphicx,epsf}
\usepackage{subfigure}
\usepackage[toc,page,header]{appendix}
\usepackage{minitoc}
\usepackage[numbers,sort&compress]{natbib}
\usepackage{color}
\newcommand{\be}{\begin{equation}}
\newcommand{\ee}{\end{equation}}
\newcommand{\ba}{\begin{eqnarray}}
\newcommand{\ea}{\end{eqnarray}}
\newcommand\tr{{\operatorname{tr}}}
\newcommand{\ignore}[1]{}
\newcommand{\otoc}[0]{\text{OTOC}}
\newcommand{\haar}[0]{\mathcal U(d)}

\newcommand{\ket}[1]{\left | {#1} \right \rangle }

\newcommand{\bra}[1]{\left \langle {#1} \right | }

\newcommand{\de}[0]{{\operatorname{d}}}

\newcommand{\expf}[1]{\mathrm{exp}\left ( {#1}\right )}

\newcommand{\parent}[1]{\left( {#1} \right)}
\newcommand{\aver}[1]{ \left\langle  {#1}  \right\rangle }

\newcommand{\spann}[1]{\mathrm{span} \parent{{#1}}}

\newcommand{\sym}{\text{sym}}
\newcommand{\pur}{\operatorname{Pur}}
\newcommand{\kerf}{\operatorname{ker}}

\def\norm#1{\Vert #1\Vert}
\def\CC{{\rm\kern.24em \vrule width.04em height1.46ex depth-.07ex
    \kern-.29em C}}
\def\P{{\rm I\kern-.25em P}}
\def\RR{{\rm
         \vrule width.04em height1.58ex depth-.0ex
         \kern-.04em R}}
\def\bbbone{{\mathchoice {\rm 1\mskip-4mu l} {\rm 1\mskip-4mu l}
{\rm 1\mskip-4.5mu l} {\rm 1\mskip-5mu l}}}
\def\bbbc{{\mathchoice {\setbox0=\hbox{$\displaystyle\rm C$}\hbox{\hbox
to0pt{\kern0.4\wd0\vrule height0.9\ht0\hss}\box0}}
{\setbox0=\hbox{$\textstyle\rm C$}\hbox{\hbox
to0pt{\kern0.4\wd0\vrule height0.9\ht0\hss}\box0}}
{\setbox0=\hbox{$\scriptstyle\rm C$}\hbox{\hbox
to0pt{\kern0.4\wd0\vrule height0.9\ht0\hss}\box0}}
{\setbox0=\hbox{$\scriptscriptstyle\rm C$}\hbox{\hbox
to0pt{\kern0.4\wd0\vrule height0.9\ht0\hss}\box0}}}}
\def\bbbz{{\mathchoice {\hbox{$\sf\textstyle Z\kern-0.4em Z$}}
{\hbox{$\sf\textstyle Z\kern-0.4em Z$}}
{\hbox{$\sf\scriptstyle Z\kern-0.3em Z$}}
{\hbox{$\sf\scriptscriptstyle Z\kern-0.2em Z$}}}}

\usepackage{hyperref}
\newtheorem{theorem}{Theorem}
\newtheorem{corollary}{Corollary}
\newtheorem{lemma}{Lemma}
\newtheorem{remark}{Remark}
\newtheorem{application}{Application}
\newtheorem{prop}{Proposition}

\begin{document}
\setcounter{secnumdepth}{3}


\title{Quantum Chaos is Quantum} 
\author{Lorenzo Leone}\email{Lorenzo.Leone001@umb.edu}
\affiliation{Physics Department,  University of Massachusetts Boston,  02125, USA}
\author{Salvatore F.E. Oliviero}
\affiliation{Physics Department,  University of Massachusetts Boston,  02125, USA}
\author{You Zhou}
\affiliation{School of Physical and Mathematical Sciences, Nanyang Technological University, 637371, Singapore}
\affiliation{Department of Physics, Harvard University, Cambridge, Massachusetts 02138, USA}
\author{Alioscia Hamma}
\affiliation{Physics Department,  University of Massachusetts Boston,  02125, USA}
\begin{abstract}
It is well known that a quantum circuit on $N$ qubits composed of Clifford gates with the addition of $k$ non Clifford gates can be simulated on a classical computer by an algorithm scaling as $\operatorname{poly}(N)\expf{k}$\cite{bravyi2016improved}. We show that, for a quantum circuit to simulate quantum chaotic behavior, it is both necessary and sufficient that $k=\Theta(N)$. This result implies the impossibility of simulating quantum chaos on a classical computer.
\end{abstract}
\maketitle

\doparttoc
\faketableofcontents

\section*{Introduction}

Quantum chaos is a certain type of complex quantum behavior that results in the exponential decay of out-of-time-order correlation functions (OTOC)\cite{kitaev2014hidden,roberts2017chaos,harrow2020separation} efficient operator spreading\cite{nahum2018operator,khemani2018spreading}, small fluctuations of the purity\cite{Oliviero2020random} and information scrambling\cite{ding2016conditional,hosur2016chaos}. All these quantities can be unified in a single framework\cite{leone2020isospectral} which shows that, in order to simulate quantum chaos, one needs at least a unitary $4$-design, that is, a set of unitary operators that reproduces up to the four moments of the Haar distribution over the unitary group $\mathcal U(d)$ in an $d$-dimensional Hilbert space. {Here, we define  quantum chaos for a quantum evolution in terms of attaining the Haar value for general multi-point OTOC, that is, the value that would be reached by a random unitary operator in $\mathcal U (d)$. We consider a subgroup of the unitary group - the Clifford group - which only reproduces up to the four-point OTOC\cite{roberts2017chaos} and thus it is not sufficient to simulate quantum chaotic evolutions.} In \cite{zhou2020single}, it was shown numerically that a Clifford circuit on a $d=2^N$-dimensional system of $N$ qubits doped by a single $T$ gate can bring a typical product state in an entangled state with the same entanglement spectrum statistics resulting from the random matrix theory for $\mathcal U(d)$. This result opens the question of whether it would be possible to simulate quantum chaos with classical resources. In a seminal paper\cite{gross2020quantum}, the authors show that an $\epsilon$-approximate $t$-design can be obtained by doping a Clifford circuit with $k=O(t^4\log^2t\log\epsilon^{-1})$ non Clifford gates. In particular, one can $\epsilon$-simulate the quantum channel that realizes a $4$-design by classical resources. This result is striking: by injecting a vanishing density $\sigma=k/N$ of non Clifford gates in a Clifford quantum circuit - as the authors say, {\em homeopathically} - one can obtain any $\epsilon$-approximate $t$-design.
Does this mean that one can simulate quantum chaos classically? The answer is no, because - as we will show - to simulate quantum chaos, the error $\epsilon$ must be exponentially small in $N$, $\epsilon =O(d^{-\alpha})$, where $\alpha$ only depends on the Haar average over the full unitary group. A corollary of the result in\cite{gross2020quantum} is that a {\em sufficient} condition to simulate quantum chaos requires $O(N)$ non Clifford resources. 

In this paper, we show that $\Theta(N)$ non Clifford resources are {\em both necessary and sufficient} to simulate quantum chaos. To this end, we explicitly compute the $8$-point OTOC and the fluctuations of the purity in a subsystem and show that a doped Clifford circuit will attain the Haar values for these quantities if and only if $\Theta(N)$ non Clifford resources are used. 
In other words, one needs more than a homeopathic dose of non Clifford gates to simulate quantum chaos. Can a classical computer simulate quantum chaos? In order to simulate a Clifford circuit with a $\Theta(N)$ non-Clifford resources, an exponential number of classical resources are needed \cite{bravyi2016improved}. Complexity-theoretic arguments \cite{bremner2011classical, Harrow2017supremacy} imply that one cannot simulate efficiently on a classical computer a quantum Clifford circuit doped with $\Theta(N)$ non-Clifford gates, and therefore, since this is necessary to simulate quantum chaos, the latter cannot be efficiently simulated on a classical computer: quantum chaos is quantum.

\section{Doped Random Quantum Clifford Circuits}\label{sec:schemdop}
Consider doped random quantum Clifford circuits $U^{(k)}$ on a system 
$\mathcal{H}=\mathbb{C}^{2\otimes N}$ of $N$ qubits of dimension $d=2^N$. The architecture of the circuit is the following: we have layers of random Clifford unitary operators on the full $\mathcal{H}$ interspersed by a single qubit gate $K_i$ applied randomly on any qubit $i$, see Fig.\ref{fig1}. As we shall see in Sec. \ref{lemma}, the positioning $i$ of the gates $K$ does not play any role. We denote by $k$ the number of gates $K$ in the circuit, also called the number of {\em layers} of the circuit, $\psi$ a pure input state for the circuit, and $\psi_U = U\psi U^\dagger$ its output. We call the quantity $\sigma=k/N$ the {\em doping} of the circuit $U^{(k)}$.

\begin{figure}[h!]
    \centering
    \includegraphics[scale=0.1]{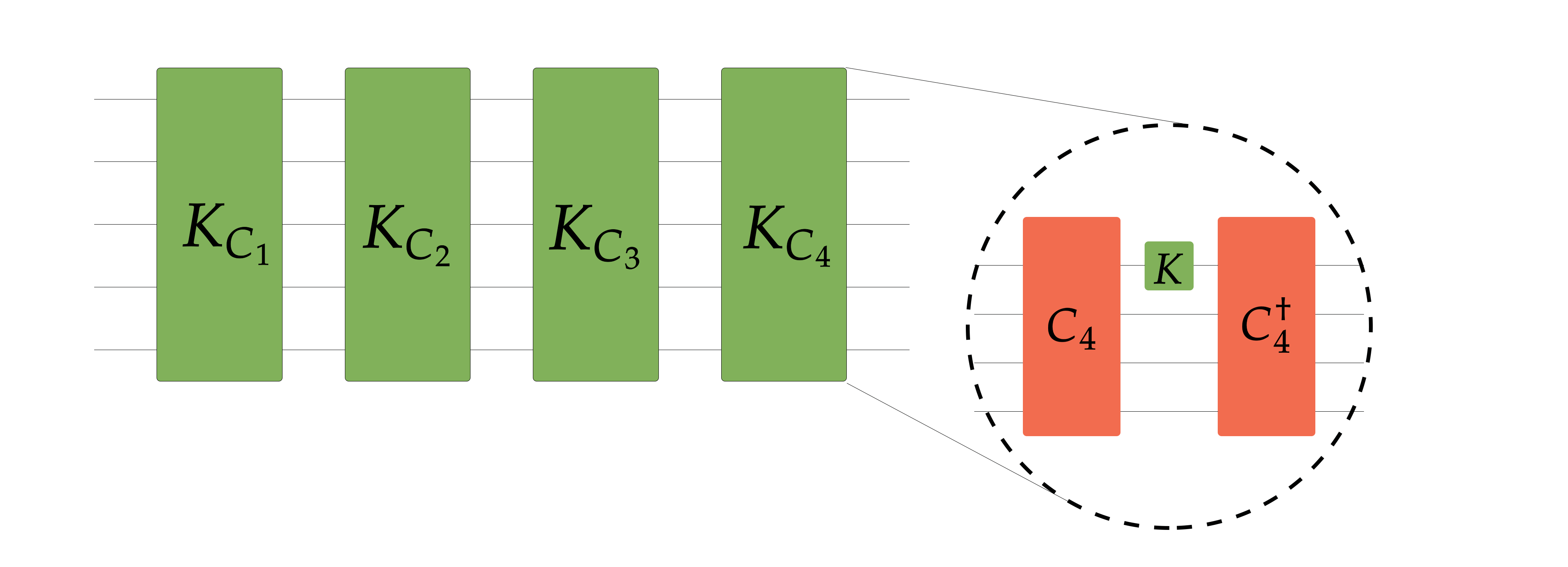}
    \caption{{\em Left}: Scheme of the $4$-Doped Clifford circuit. {\em Right}: Detail of $K_{C_{4}}$, a unitary single-qubit non Clifford gate $K$ evolved adjointly by a Clifford circuit $C_{4}$. Note that the set formed by these circuits is equivalent to the set of doped Clifford circuits, i.e circuits composed by Clifford unitaries $C_{i}$ interspersed with single-qubit non Clifford gates $K_{i}$.}
    \label{fig1}
\end{figure}

We denote by $x$ a set of unitary operators, e.g. $x= \mathcal U(d), \mathcal C(d)$ the unitary and Clifford group, respectively, on $\mathcal{H}$. For $k=0$, the circuit is just a Clifford circuit, $U^{(0)}\in\mathcal C(d)$. 
The Haar average on these sets will be denoted by $\langle\cdot\rangle_{U\in x}$. We define the $(x,t)$-fold channel as
\ba
\Phi_{x}^{(t)}(\mathcal{O}):=\aver{\mathcal{O}_{U}}_{U\in x}
\ea
where $\mathcal{O}\in\mathcal{B}(\mathcal{H}^{\otimes t})$ and $\mathcal{O}_{U}\equiv U^{\otimes t}\mathcal{O}U^{\dag\otimes t}$. 
Averaging over $\mathcal C(d)$ for a circuit $U^{(k)}$ with $k$ layers involves averaging over $k$ independent Clifford groups; in the following we define this set of circuits as $\mathcal{C}_k$. The $(\mathcal{C}_k,4)$-fold channel is
\be\label{groupav}
\Phi^{(4)}_{\mathcal C_k} (\mathcal O)= \left\langle \mathcal{O}_{U}\right\rangle_{U\in\mathcal{C}_k} \equiv \left\langle C_{k}^{\otimes 4}K_{i_k}^{\otimes 4}\dots C_{1}^{\otimes 4}K_{i_{1}}^{\otimes 4}C_{0}^{\otimes 4}\mathcal{O}C_{0}^{\dag\otimes 4}K_{i_1}^{\dag\otimes 4}C_{1}^{\otimes 4}\dots K_{i_{k}}^{\dag\otimes 4}C_{k}^{\dag\otimes 4}\right\rangle_{C_1\ldots C_k\in\mathcal{C}(d)}
\ee
Notice that the above average over $\mathcal C(d)$ is the same thing - because of the left/right invariance of group averages - than the average over circuits of the type sketched in Fig.\ref{fig1}.

Quantum chaos can be defined as an appropriate form of the butterfly effect\cite{roberts2016bound}: an exponential (in $N$) decay of the OTOCs defined as
\be
\otoc_{8}(U):=d^{-1}\tr\left(AB_{U}CD_{U}A^\dag D^\dag_{U}C^\dag B^\dag_{U}\right)
\label{8otocdef}
\ee
so that the OTOCs adhere to the value of the OTOCs obtained by Haar-random $U$ on the unitary group scaling with $d^{-4}$, while other ensembles, like the Clifford group, feature a scaling of $d^{-2}$\cite{roberts2017chaos}. It is immediate to see that, in order to distinguish the two types of scaling, one needs an $\epsilon =O(d^{-4})$. As  $2t$-OTOCs are probes of $t$-designs, an $8$-point OTOC is a probe of a $4$-design, and therefore a quantum chaotic channel needs to have a frame potential exponentially close to that of the Haar measure on $\mathcal U(d)$. 

A related measure of chaos\cite{leone2020isospectral} is given by the fluctuations of the purity of the reduced density matrix to a subsystem
 \ba\label{flucpurdef1}
 \Delta_{x}\pur(\psi_U)_A:= \left\langle (\pur(\psi_U)_A-\langle \pur(\psi_U)_A\rangle_{U\in x})^2\right\rangle_{U\in x}
\ea
This quantity is related to the emergent irreversibility in closed quantum systems\cite{chamon2014emergent} and to both $4$-designs and OTOCs. In Sec. \ref{flucpurity}, we show that the purity fluctuations are exponentially small for every doping (including no doping, $k=0$) of the random Clifford circuits and thus also to distinguish the fluctuations of the purity one needs an exponentially small error $\epsilon$. We ask the question: what is the necessary and sufficient number $k$ of non Clifford gates $K$ for $U\in\mathcal{C}_{k}$ to simulate quantum chaos?

The main goal of this paper is to show that, for $\mathcal{C}_{k}$ to reproduce the Haar-unitary values of the probes Eqs. \eqref{8otocdef} and \eqref{flucpurdef1}, $\Theta(N)$ non Clifford resources are both necessary and sufficient. We will prove it in the next sections by explicitly computing these two quantities. Here, we want to make some more general considerations. Given a probe to quantum chaos defined as $\mathcal{P}_t(U):=\tr(T^{(t)}\mathcal{O}_1 U^{\otimes t}\mathcal{O}_{2}U^{\dag\otimes t})$, see\cite{leone2020isospectral,Oliviero2020random}, we can establish the following
\begin{prop}
Let $\mathcal{P}_{t}(U)$ a probe of quantum chaos of order $t$. If the number $k$ of non Clifford gates in the doped Clifford circuit $U\in \mathcal{C}_k$ is $k=O((\alpha+t) Nt^{4}\log^2(t))$, then:
\be
\delta_{\mathcal{P}}^{(k)}\equiv\left|\aver{\mathcal{P}(U)}_{U\in\mathcal{C}_k}-\aver{\mathcal{P}(U)}_{U\in\mathcal{U}(d)}\right|\le O(d^{-\alpha})
\ee
where $\alpha$ is given by the following relation $\aver{\mathcal{P}(U)}_{U\in\mathcal{U}(d)}=O(d^{-\alpha})$.
\end{prop}
\begin{proof} The proof is straightforward from  the result in \cite{gross2020quantum}. In \cite{leone2020isospectral}, we proved that the generic probe $\mathcal{P}_t(U)$ to quantum chaos can be written as $\mathcal{P}_t(U)=\tr(T^{(t)}\mathcal{O}_1 U^{\otimes t}\mathcal{O}_{2}U^{\dag\otimes t})$ where $\mathcal{O}_1,\mathcal{O}_2, T^{(t)}\in\mathcal{B}(\mathcal{H}^{\otimes t})$, including the $2t$-point OTOC which characterize $t$-designs\cite{roberts2017chaos}. Then, the following inequality holds:
\be
\delta_{\mathcal{P}}^{(k)}\le \norm{T^{(t)}\mathcal{O}_1}_{\infty}\norm{\mathcal{O}_2}_{1}\norm{\Phi^{(t)}_{\mathcal{C},k}(\cdot)-\Phi^{(t)}_{\mathcal{U}(d)}(\cdot)}_{\diamond}\le d^t  \norm{T^{(t)}\mathcal{O}_1}_{\infty}\norm{\mathcal{O}_2}_{\infty}\norm{\Phi^{(t)}_{\mathcal{C},k}(\cdot)-\Phi^{(t)}_{\mathcal{U}(d)}(\cdot)}_{\diamond}
\ee
where we bounded $\norm{\mathcal{O}_2}_{1}\le    \norm{\mathcal{O}_2}_{\infty}\norm{\bbbone^{\otimes t}}_{1}=d^t\norm{\mathcal{O}_2}_{\infty}$. Thus if $k=O((\alpha+t) Nt^{4}\log^2(t))$ then $\delta^{(k)}_{\mathcal{P}}=O(d^{-\alpha})$.\end{proof}

As we will show in the following sections, the error $O(d^{-\alpha})$ is required to have OTOCs and fluctuations of the purity attain the unitary-Haar values. It follows that injecting  $O(N)$ non Clifford resources into a Clifford circuit is sufficient to obtain quantum chaos.

As we stated above, the necessary (together with the sufficient) condition will follow from direct calculations. It is important at this point to make some remarks about the value of $t$. One wonders if it is enough to consider $8$-OTOCs to reveal quantum chaos, or if sometimes it should be necessary to use higher order OTOCs. From the point of view of the above proposition, it is clear that $O(N)$ non Clifford resources are sufficient to obtain any OTOC  with an exponentially good $O(d^{-\alpha})$ approximation.  Once one has proved the necessary condition for $t=4$, it will also hold for any $t>4$ design, as an approximate $t$-design is necessarily a $t^{\prime}$ approximate design, for all $t^{\prime}<t$. In other words, polynomials of degree four are all that takes to reveal quantum chaos. {To see this, notice that the $4m-$OTOCs defined as 
\be
\otoc_{4m}(U):=d^{-1}\tr[(A_{1}B_{1}^{U}A_{2}B_{2}^{U}\cdots A_{m}B_{m}^{U})A_{1}^{\dag}(A_{1}B_{1}^{U}A_{2}B_{2}^{U}\cdots A_{m}B_{m}^{U})]
\label{4motocdef}
\ee
  reduces,  for $A_{3},\dots, A_m,B_{3},\dots B_{m}=\bbbone$, to the $8-$OTOC in Eq. \eqref{8otocdef}; therefore $k=\Omega(N)$ are necessary to obtain any $4m-$OTOC with an exponentially good approximation.}

\section{Main Theorem}\label{sec:main}
From the technical point of view, the main result of this paper is the exact calculation of the fourth moment of the output of a $k$-doped random Clifford circuit for a generic operator $\mathcal{O}\in\mathcal{B}(\mathcal{H}^{\otimes 4})$:
\begin{theorem}\label{th1}
Let $\mathcal{O}\in\mathcal{B}(\mathcal{H}^{\otimes 4})$ be a bounded operator,  $U\in \mathcal{C}_k$  a $k$-doped Clifford circuit; then the $(\mathcal{C}_k,4)$-fold channel for the $k$-doped Clifford circuit reads
\be
\Phi^{(4)}_{\mathcal{C}_k} (\mathcal O)=\sum_{\pi,\sigma\in\mathcal{S}_4}\left[\parent{(\Xi^{k})_{\pi\sigma}Q+\Gamma_{\pi\sigma}^{(k)}}c_{\pi}(\mathcal{O})+\delta_{\pi\sigma}b_{\pi}(\mathcal{O})\right]T_{\sigma}
\label{t1}
\ee
where
\be
Q=\frac{1}{d^2}\sum_{P\in\mathcal{P}(2^N)}P^{\otimes 4};\qquad\; Q^{\perp}=\bbbone^{\otimes 4}-Q
\ee 
and $\mathcal{P}(2^N)$ is the Pauli group on $N$ qubits;
$T_{\pi}$ are permutation operators corresponding to $\pi\in\mathcal{S}_{4}$, then $\Xi^{k}$ is the $k$-matrix power of the matrix $\Xi$, whose components read 
\be
\Xi_{\sigma\pi}\equiv\sum_{\tau\in\mathcal{S}_4}\left[W^{+}_{\pi\tau}\tr\left(T_{\sigma}K^{\otimes 4}QK^{\dag\otimes 4}QT_{\tau})-W^{-}_{\pi\tau}\tr(T_{\sigma}K^{\otimes 4}QK^{\dag\otimes 4}Q^{\perp}T_{\tau}\right)\right]\label{xi}
\ee
with
\ba
\Gamma^{(k)}_{\pi\sigma}&\equiv&\sum_{\tau\in\mathcal{S}_4}\Lambda_{\pi\tau}\sum_{i=0}^{k-1}(\Xi^i)_{\tau\sigma}\\
\Lambda_{\pi\tau}&\equiv&\sum_{\sigma\in\mathcal{S}_4}W^{-}_{\pi\sigma}\tr(T_{\tau}K^{\otimes 4}QK^{\dag\otimes 4}Q^{\perp}T_{\sigma})
\ea
and the information about the operator $\mathcal O$ is all contained in the coefficients
\ba\label{bc_coeff}
c_{\pi}(\mathcal{O})&\equiv&\sum_{\sigma\in\mathcal{S}_4}[W^{+}_{\pi\sigma}\tr(\mathcal{O}QT_{\sigma})-W^{-}_{\pi\sigma}\tr(\mathcal{O}Q^{\perp}T_{\sigma})]\label{ccoeff}\\
b_{\pi}(\mathcal{O})&\equiv&\sum_{\sigma\in\mathcal{S}_4}W^{-}_{\pi\sigma}\tr(\mathcal{O}Q^{\perp}T_{\sigma})\label{bcoeff}
\ea
where $W^{\pm}_{\pi\sigma}$ are the generalized Weingarten functions for the Clifford group, introduced and discussed in App. \ref{Cliffhaar}. 
\end{theorem}
The proof of the theorem can be found in App. \ref{proofmaintheorem}.

For many purposes, it is important to know to what  $\Phi^{(4)}_{\mathcal{C}_k}(\cdot)$ converges in the limit of infinite layers. Without substantial loss of generality, we consider the case of the non Clifford resources given by phase gates $P_\theta$ with $\theta\ne\pi/2$. We can thus establish the application:
\begin{application}\label{thconvergence} 
For $K=P_{\theta}\equiv\ket{0}\bra{0}+e^{i\theta}\ket{1}\bra{1}$, where $\{\ket{0},\ket{1}\}$ is the single qubit computational basis, and for any $\theta\neq\pm \pi/2$ the $(\mathcal{C}_k,4)$-fold channel equals the $(\mathcal{U}(d),4)$-fold channel  in the limit $k\rightarrow\infty$ is
\be
\lim_{k\rightarrow\infty}\Phi^{(4)}_{\mathcal{C}_k} (\mathcal O)=\Phi^{(4)}_{\mathcal U(d)} (\mathcal O)
\ee
\end{application}
The proof can be found in App. \ref{proofthconvergence}. Note that the above result can be also seen as a consequence of the results in \cite{harrow2009random}.

In the next sections, we apply these theorems to calculating the $8$-point OTOCs and fluctuations of subsystem purity to find how these quantities approach the Haar-average on $\mathcal U (d)$ with $k$.

\section{The $8$-point OTOC}\label{sec:8otoc}
Consider four non-identity and non-overlapping Pauli operators $A,B,C,D\in\mathcal{P}(d)$. Then consider the unitary evolution of $A_{U}=UAU^{\dag}$ in the Heisenberg picture and define an $8$-point Out of Time Order Correlator (OTOC) as\cite{roberts2017chaos}, defined in Eq.\eqref{8otocdef}
\be
\otoc_{8}(U):=d^{-1}\tr(AB_UCD_UAD_UCB_U)
\label{8otocdef2}
\ee
We are interested in taking the twirling of the $8$-point OTOC for a $k$-doped Clifford circuit, in order to find a necessary and sufficient condition for the exponential decay of the OTOC. Thanks to Theorem \ref{th1} we obtain
\begin{application}\label{c3} Let $K\equiv T$ the single qubit $T$-gate, then the average of the $8$-point OTOC over the $k$-doped Clifford circuit reads
\ba
\aver{\otoc_{8}(U)}_{U\in\mathcal{C}_k}&=&\frac{5d^2}{(d^2-1)(d^2-4)(d^2-9)}-(f^{-})^{k}\frac{d(d^2+4d+6)}{6(d^2-1)(d+2)(d+3)}\nonumber\\&+&(f^{+})^{k}\frac{d(d^2-4d+6)}{6(d^2-1)(d-2)(d-3)}+   \parent{\frac{f^{+}+f^{-}}{2}}^{k}\frac{4d^2}{3(d^2-1)(d^2-4)}
\label{otocfinalresult}
\ea
where $f^{\pm}\equiv f^{\pm}_{\pi/4}=\frac{3d^2\pm3d-4}{4(d^2-1)}=\frac{3}{4}+\Theta(d^{-1})$, where $f^{\pm}_{\theta}$ are defined in Eq.\eqref{fpmtheta}.
\end{application}

\begin{proof}
Starting from Eq.\eqref{8otocdef2} we can write the $8$-point OTOC for $U$ as
\be
\otoc_{8}(U)=d^{-1}\tr(T_{(1432)}(A\otimes C\otimes A\otimes C)U^{\otimes 4}(B\otimes D\otimes D\otimes B)U^{\dag\otimes 4})
\label{otocpermutation}
\ee
Taking the average over the $k$-doped Clifford $U\in \mathcal{C}_{k}$ we have
\be
\aver{\otoc_8(U)}_{U\in\mathcal{C}_k}=d^{-1}\tr\left(T_{(1432)}\left(A\otimes C\otimes A\otimes C\right)\aver{\left(B\otimes D\otimes D\otimes B\right)_{U}}_{U\in\mathcal{C}_k}\right)
\ee
from the latter equation, the calculation is a straightforward, but tedious application of Theorem \ref{th1}. \end{proof}

The following corollary explicitly shows the difference in the scaling of the $8$-point OTOC for a pure Clifford circuit and a universal circuit. As we shall see there is a marked difference in these scalings. As a direct consequence of Theorem \ref{th1} and Application \ref{thconvergence}, we obtain the following 
\begin{corollary}
Taking the average for $U\in\mathcal{U}(d)$ and $C\in\mathcal{C}(d)$ of the $8$-point OTOC, one gets
\ba
\aver{\otoc_8(U)}_{U\in\mathcal{U}(d)}&=&\frac{5d^2}{(d^2-1)(d^2-4)(d^2-9)}=\frac{5}{d^4}+\Theta(d^{-6})\label{otochaarexp}\\
\aver{\otoc_8(C)}_{C\in\mathcal{C}(d)}&=&\frac{d^2}{d^4-5d^2+4}=\frac{1}{d^2}+\Theta(d^4)\label{otocclifford1}
\ea
\end{corollary}
\begin{proof} The proof of Eq.\eqref{otochaarexp} can be obtained from Eq.\eqref{otocfinalresult} in the limit $k\rightarrow\infty$, in virtue of Theorem \ref{thconvergence},  while Eq.\eqref{otocclifford1} can be obtained from Eq.\eqref{otocfinalresult} setting $k=0$.\end{proof}

With the following statement, we give the necessary and sufficient condition for the number of non Clifford gates needed to precisely simulate the behavior of the $8$-point OTOC, and thus to simulate quantum chaos.

\begin{corollary}\label{c5}
Iff $k=\Theta(\log d)$, then
\be
\delta_{\otoc}^{(k)}\equiv\left|\aver{\otoc_{8}(U)}_{U\in\mathcal{C}_k}-\aver{\otoc_{8}({U})}_{U\in\mathcal{U}(d)}\right|=\Theta(d^{-4})
\ee
\end{corollary}

\begin{proof}
Taking the difference in absolute value between Eq.\eqref{otocfinalresult} and \eqref{otochaarexp} we get
\ba
\delta_{\otoc}^{(k)}&=&-(f^{-})^{k}\frac{d(d^2+4d+6)}{6(d^2-1)(d+2)(d+3)}+(f^{+})^{k}\frac{d(d^2-4d+6)}{6(d^2-1)(d-2)(d-3)}\nonumber\\&+&   \parent{\frac{f^{+}+f^{-}}{2}}^{k}\frac{4d^2}{3(d^2-1)(d^2-4)}
\ea
Taking the asymptotic limit for $d\rightarrow\infty$ up to $\Theta(d^{-4})$
\be
\delta_{\otoc}^{(k)}=\frac{1}{d^2}\left|\frac{k-3}{3}\right|\parent{\frac{3}{4}}^{k}+\Theta(d^{-4})
\ee
from here it's easy to see that one has the following condition
\be
\left(\frac{3}{4}\right)^{k}k=\Theta(d^{-2})\iff k=\Theta(\log d)
\ee
which leads to the desired result. \end{proof}

\section{Purity and its Fluctuations} \label{flucpurity}
In this section, we compute the fluctuations of a subsystem purity Eq.\eqref{flucpurdef1} for the output of the $k$-doped Clifford circuit $U\in\mathcal{C}_k$. To this end, we first
apply Theorem \ref{th1} to calculate the average of the fourth tensor power of a pure state $\psi$, namely 
$\Phi^{(4)}_{\mathcal{C}_k}(\psi^{\otimes 4}) $.


\begin{application}\label{application1}
The $(\mathcal{C}_k,4)$-fold channel of a pure state $\psi^{\otimes 4}\in\mathcal{B}(\mathcal{H}^{\otimes 4})$ reads
\ba\label{maintheorem}
\Phi^{(4)}_{\mathcal{C}_k}(\psi^{\otimes 4}) =a_{k}Q\Pi_{\sym}^{(4)}+b_{k}\Pi_{\sym}^{(4)}
\ea
where $\Pi_{\sym}^{(4)}$ is the  projector onto the completely symmetric subspace of the permutation group $\mathcal S_4$ and $D_{\sym}=\tr(\Pi_\sym^{(4)})$. The coefficients $a_k,b_k$ are given by
\ba
a_{k}&\equiv&\left(\frac{\tr(\psi^{\otimes 4}Q)}{D^+}-\frac{\tr(\psi^{\otimes 4}Q^{\perp})}{D^{-}}\right)\left(\frac{c_Q}{D^+}-\frac{c_{QQ^\perp}}{D^-}\right)^{k}\nonumber\\
b_{k}&\equiv&\frac{\tr(\psi^{\otimes 4}Q^{\perp})}{D^{-}}+\frac{c_{QQ^\perp}}{D^-}\sum_{i=0}^{k-1}a_i
\label{coefficientmaintheorem}
\ea
with
\be
c_{Q}\equiv\tr(K^{\otimes 4}QK^{\dag\otimes 4}Q\Pi_{\sym}^{(4)}), \quad c_{QQ^\perp}\equiv\tr(K^{\otimes 4}QK^{\dag\otimes 4}Q^{\perp}\Pi_{\sym}^{(4)})
\ee
\end{application}
The proof can be found in App. \ref{proofcorollary1}. 
The evaluation of Eq.\eqref{maintheorem} becomes particularly simple if the gate $K$ is a $P_\theta$-gate:
\begin{application}\label{app4} 
If the single qubit gate $K$ is the $P_\theta$-gate, the coefficients $c_{Q}, c_{QQ^\perp}$ read: 
\ba
c_{Q}&=&\frac{(d+2)(4+7d+(4+d)\cos(4\theta))}{48}\\
c_{QQ^\perp}&=&\frac{(d+2)(d+4)}{24}\sin^{2}(2\theta)
\ea
Then for any $k$ we can write the coefficients $a_k,b_k$ as
\ba
a_{k}&=&\frac{24 }{(d^2-1)(d+2)(d+4)}\left(\frac{d(d+3)}{4}\tr(\psi^{\otimes 4} Q)-1\right)(f_{\theta}^{-})^{k}\nonumber\\
b_{k}&=&\frac{1}{D_{\sym}}+\frac{24}{(d^2-1)(d+2)(d+4)}\left(\frac{4}{d(d+3)}-\tr(\psi^{\otimes 4}Q)\right)(f_{\theta}^{-})^{k}\label{kcoefficienttheta}
\ea
where $(f_{\theta}^{-})$ is defined in Eq.\eqref{fpmtheta}; note that $(f_{\theta}^{-})<1$ unless $\theta=\pm \pi/2$, i.e unless $P_\theta=S$ the $S$-gate $\in \mathcal{C}(d)$.
\end{application}
See App. \ref{calculationcqcqqp} for the proof.
\begin{corollary}\label{corollaryptheta}
For any $\theta\neq \pm \pi/2$
\be\label{Tgatetheorem}
\lim_{k\rightarrow\infty}\Phi^{(4)}_{\mathcal{C},k}(\psi^{\otimes 4})=\Phi^{(4)}_{\mathcal{U}(d)}(\psi^{\otimes 4})
\ee
\end{corollary}

\begin{proof} The proof follows directly from Application \ref{thconvergence}; here we give an alternative version: setting $f_{\theta}<1$ in Eq.\eqref{kcoefficienttheta} and taking the limit $k\rightarrow\infty$ one gets
\ba
\lim_{k\rightarrow\infty}a_{k}&=&0\\
\lim_{k\rightarrow\infty}b_{k}&=&D_{\sym}^{-1}
\ea
Now, since the fourth tensor power of $\psi_U$ averages to - see Eq.\eqref{haaraverageofstate} in App. \ref{puraveragesec} for a proof -
\be
\Phi^{(4)}_{\mathcal{U}(d)}(\psi^{\otimes 4})=\frac{\Pi_{\sym}^{(4)}}{D_{\sym}}
\ee
then by Application \ref{application1} the proof is complete.
 \end{proof}

In what follows, we  calculate the purity and its fluctuations in a bipartite Hilbert space for the output state of a $k$-doped Clifford circuit, calculated above in Eq.\eqref{maintheorem}. Consider then a bipartition of the $N$-qubit system $\mathcal{H}= \mathcal{H}_A\otimes \mathcal{H}_B$ with $\mathcal H_{A(B)}= \mathbb{C}^{2\otimes N_{A(B)}}$, $N_{A}+N_{B}=N$ and $d_{A(B)}=2^{N_{A(B)}}$. The purity of a marginal state $\psi_A =\tr_B\psi \in \mathcal B(\mathcal H_A)$ is given by
\be
\pur(\psi_A):=\tr(\psi_{A}^2)
\ee
The averages over Unitary and Clifford group for the purity of the output $\psi$ of a random quantum circuits are the same, namely
\be
\aver{\pur(\psi_U)_A}_{U\in\haar}=\aver{\pur(\psi_C)_A}_{C\in \mathcal C(d)}=\frac{d_A+d_B}{d_Ad_B+1}
\label{purityaverage}
\ee
This is a consequence of $\mathcal C(d)$ being a $3$-design\cite{webb2016clifford,Zhu2017multiqubit} (in fact, being a $2$-design is sufficient), 
 see App. \ref{puraveragesec} for a proof. Notice that the average purity does not depend on the input state.
 
 The fluctuations of the purity for the set $x$ are defined as
 \ba\label{flucpurdef2}
 \Delta_{x}\pur(\psi_U)_{A}:= \langle (\pur(\psi_U)_A-\langle \pur(\psi_U)_A\rangle_{U\in x})^2\rangle_{U\in x}
 \ea
 Since the fluctuations involve the fourth moment of the Haar measure, the fluctuations for $\mathcal U(d), \mathcal C(d)$ are expected to be different. We have indeed, for 
 $d_A=d_B=\sqrt{d}$
\ba
\Delta_{\haar}\pur(\psi_U)_{A}&=& \frac{2(d-1)^2}{(d+1)^2(d+2)(d+3)}=\Theta(d^{-2})\\
\Delta_{\mathcal{C}(d)}\pur(\psi_C)_{A}&=&\frac{(d-1)[d(d+1)\tr(Q\psi^{\otimes 4})-2]}{(d+1)^2(d+2)}=\begin{cases}\Theta(d^{-1}), \quad  \psi=\ket{0}\bra{0}^{\otimes N}\\
\Theta(d^{1-\log_2 5}), \quad  \psi=\otimes_i\psi_i, \; \psi_i \; \text{random}
\end{cases}
\label{cliffordfluct}
\ea
This result is a consequence of 
Application \ref{application1} and Corollary \ref{corollaryptheta}. Notice that while the fluctuations of the purity for the unitary group again do not depend on the initial state, those for the Clifford group do. For $\psi $ being any other stabilizer state different from $\ket{0}\bra{0}^{\otimes N}$, the formula would not change thanks to the left/right invariance of the Haar measure over groups. Notably, starting from completely factorized states, there is a marked difference whether the initial state $\psi$ is a stabilizer state or a random product state.

\begin{lemma}\label{lemmafluc}
The fluctuations of the purity in the $k$-doped Clifford circuit, for $d_A=d_B=\sqrt{d}$ and $\psi=\ket{0}\bra{0}^{\otimes N}$,   are
\be
\Delta_{\mathcal{C}_k}\pur(\psi_U)_{A}=\frac{(d-1)^2}{(d+1)^2(d+2)(d+3)}(2+(d+1)(f_{\theta}^{-})^{k})
\label{fluctclifforddoped}
\ee
where $f_\theta^{-}$ is defined in Eq.\eqref{fpmtheta}.
\end{lemma} 

The proof can be found in App. \ref{prooflemmafluc}.

\begin{remark}
For the undoped, $k=0$, pure Clifford circuit, one finds
\be 
\aver{\pur^2(\psi_{C})_A}_{C\in\mathcal{C}(d)}=\frac{5d+1}{(d+1)(d+2)}
\ee 
Notice that, in the large $d$ limit, 
\ba
\aver{\pur^2(\psi_{C})_A}_{C\in\mathcal{C}(d)}&=& \frac{5}{d}+\Theta(d^{-2})\\
\aver{\pur^2(\psi_{U})_A}_{U\in\mathcal{U}(d)}&=&\frac{4}{d}+\Theta(d^{-2})
\ea
and thus have the same order. However, the next corollary shows  that - because of an exact cancellation - the fluctuations are very different in scaling with $d$.
\end{remark}

\begin{corollary} The fluctuations of the purity, for $d_A=d_B=\sqrt{d}$ and $\psi=\ket{0}\bra{0}^{\otimes N}$ for the Clifford circuit are
\ba 
\Delta_{\mathcal{C}(d)}\pur(\psi_C)_{A}&=&\frac{(d-1)^2}{(d+1)^2(d+2)}=\Theta(d^{-1})\label{eq41}\\
\Delta_{\mathcal{U}(d)}\pur(\psi_U)_{A}&=& \frac{2(d-1)^2}{(d+1)^2(d+2)(d+3)}=\Theta(d^{-2})\label{eq42}
\ea 
\end{corollary}
\begin{proof}
Eq.\eqref{eq41} follows immediately from Lemma \ref{lemmafluc}  by setting $k=0$, while Eq.\eqref{eq42} can be found in \cite{Hammalungo_2012}. \end{proof}
\begin{corollary}\label{iffpurity}
For any $\theta\neq \pm \pi/2$, for $d_{A}=d_{B}=\sqrt{d}$ and $\psi=\ket{0}\bra{0}^{\otimes N}$, iff $k=\Theta(\log d)$
\be
\delta_{\pur}^{(k)}\equiv |\Delta_{\mathcal{C}_k}\pur(\psi_U)_{A}-\Delta_{\mathcal{U}(d)}\pur(\psi_U)_{A}|=\Theta(d^{-2})
\ee
\end{corollary}
\begin{proof} From Eq.\eqref{fluctclifforddoped} one has
\be
\delta_{\pur}^{(k)}=\frac{(d-1)^2}{(d+1)(d+2)(d+3)}(f_{\theta}^{-})^{k}
\ee
for $k=0$ this difference is $\Theta(d^{-1})$, then in order to get $O(d^{-2})$ one needs to have
\be
f_{\theta}^{k}=\Theta(d^{-1})\iff k=\Theta(\log d)
\ee
moreover, note that the rate of convergence is dictated by $f_{\theta}$, which reaches its minimum value for $\theta=\pi/4$, that is the $T$-gate, cfr. Eq.\eqref{fpmtheta}.\end{proof}

\begin{lemma}\label{lemma2}
The fluctuations of the purity for a $k$-doped Clifford circuit, for $d_A=d_B=\sqrt{d}$ and $\psi$ be a random product state read
\be
\Delta_{\mathcal{C}_k}\pur(\psi_U)_{A}=\frac{2(d-1)^2}{(d+1)^2(d+2)(d+3)}+\frac{(d-1)(d^{(2-\log_2 5)}(d-3)-4)(f^{-}_{\theta})^{k}}{(d+1)(d+2)(d+3)}
\ee
\end{lemma}
\begin{proof} The proof is straightforward and is left to the interested reader: by plugging \eqref{randomproductstate} into Eqs. \eqref{coefficientmaintheorem} and \eqref{maintheorem} and using Eq.\eqref{purfluctdoped} the calculation follows easily.\end{proof}

\begin{corollary}
The fluctuations of the purity for a non-doped Clifford circuit, for $d_A=d_B=\sqrt{d}$ and $\psi$ a random product state, are
\be
\Delta_{\mathcal{C}(d)}\pur(\psi_U)_{A}=\frac{(d-1)[d(d+1)d^{1-\log_{2}5}-2]}{(d+1)^2(d+2)}
\ee
\end{corollary}
\begin{proof} This result is obtained from Lemma \ref{lemma2} setting $k=0$.\end{proof}
\begin{remark}
The hypothesis $d_A=d_B=\sqrt{d}$ only simplifies the displayed formulas and the related considerations but does not change the general behavior.  For instance, in the case   $d_B \gg d_A \gg 1$ one has - after a lengthy calculation, 
\be 
\Delta_{\mathcal{C}_k}\pur(\psi_U)_{A}=\frac{2(d^2-d_{A}^2)(d^2_{A}-1)}{(d+1)^2(d+2)(d+3)d^2_{A}}+\frac{(d^2-d_{A}^2)(d_{A}^2-1)f^{k}_{\theta}(d(d+3)\tr\parent{Q\psi^{\otimes 4}}-4)}{(d-1)(d+1)(d+2)(d+3)d^2_{A}}
\ee 
\end{remark}



\section*{Conclusions and Outlook} 
In this paper, we showed that in a random Clifford circuit with $N$ qubits,  $\Theta(N)$ non Clifford gates are both necessary and sufficient to simulate quantum chaos. As a consequence, quantum chaos cannot be efficiently simulated on a classical computer, as the cost for simulating such circuits is exponential in the non Clifford resources.

In perspective, there are several open questions. One could generalize many of these results by proving that an $\epsilon$-approximate $2t$-OTOC characterizes an $\epsilon$-approximate $t$-design. Although the scaling is fixed to be $\Theta(N)$, the actual number of non Clifford resources is undetermined and it would be of practical importance in obtaining approximate $t$-designs with a noisy, intermediate-scale quantum computer. One could thus study the optimal arrangement of non Clifford resources. A related question is the onset of irreversibility in a closed quantum system in the sense of entanglement complexity\cite{chamon2014emergent} is driven by the doping of a Clifford circuit. Similarly, it would be interesting to show how the entanglement spectrum statistics converges with the doping\cite{zhou2020single}.

\section*{Acknowledgments} We acknowledge support from  NSF award number 2014000. The authors thank Claudio Chamon for important and enlightening discussions and comments.



\appendix

\section{Haar averages: Unitary vs Clifford group\label{haarintegration}}
In this section, we are going to display the explicit formula to average over the full unitary group and the full Clifford group without going into the group theoretic details. See \cite{collins2006integration,collins2003moments} for the Haar integration and \cite{roth2018recovering,zhu2016clifford} for the Clifford integration formula.
\subsection{Clifford group average}\label{Cliffhaar}
Starting from the result about the $4$-th moment of the  Haar average over the Clifford group in\cite{roth2018recovering} we are going to prove a useful lemma. 
\begin{lemma}
Let $\mathcal{O}\in\mathcal{B}(\mathcal{H}^{\otimes 4})$,  the integration formula for the Clifford group reads
\be
\aver{\mathcal{O}_{C}}_{C\in\mathcal{C}(d)}=\int_{\mathcal{C}(d)}\de C C^{\dag\otimes 4}\mathcal{O}C^{\otimes 4}=\sum_{\pi,\sigma\in\mathcal{S}_{4}}W^{+}_{\pi\sigma}\tr(\mathcal{O}QT_\pi)QT_{\sigma}+W^{-}_{\pi\sigma}\tr(\mathcal{O}Q^{\perp}T_{\pi})Q^{\perp}T_{\sigma}
\label{cliffordweing}
\ee
where $Q=\frac{1}{d^2}\sum_{P\in\mathcal{P}(d)}P^{\otimes 4}$ and $Q^{\perp}=\bbbone^{\otimes 4}-Q$, while $W^{\pm}_{\pi\sigma}$ are the generalized Weingarten functions, defined as
\be
W^{\pm}_{\pi\sigma}=\sum_{\substack{\lambda\vdash 4\\D^{\pm}_{\lambda}\neq 0}}\frac{d_{\lambda}^2}{(4!)^2}\frac{\chi^{\lambda}(\pi\sigma)}{D^{\pm}_\lambda}
\label{generalizedWeingarten}
\ee
here $\lambda$ labels the irreducible representations of the symmetric group $\mathcal{S}_4$, $\chi^{\lambda}(\pi\sigma)$ are the characters of $\mathcal{S}_4$, $d_\lambda$ is the dimension of the irreducible representation $\lambda$, $D_{\lambda}^{+}=\tr(QP_{\lambda})$ and $D_{\lambda}^{-}=\tr(Q^{\perp}P_{\lambda})$ where $P_\lambda$ are the projectors onto the irreducible representations of $\mathcal{S}_{4}$ and finally $T_{\sigma}$ are permutation operators corresponding to the permutation $\sigma\in \mathcal{S}_4$. 
\end{lemma}
\begin{proof}
The projectors onto the irreducible representations of $\mathcal{S}_4$ read
\be
\Pi_\lambda^{(4)}=\frac{d_\lambda}{4!}\sum_{\tau\in\mathcal{S}_{4}}\chi^{\lambda}(\tau)T_{\tau}
\ee
Starting from the integration formula $(32)$ in \cite{roth2018recovering} we have
\ba
\aver{\mathcal{O}_C}_{C\in\mathcal{C}(d)}\hspace{-0.5cm}&&=\frac{1}{(4!)^2}\sum_{\lambda \vdash 4, l(\lambda)\le d} d_{\lambda}^{2}\sum_{\sigma \in\mathcal{S}_{4}}\parent{\frac{1}{D_{\lambda}^{+}}\tr(\mathcal{O}QT_{\sigma})Q+\frac{1}{D_{\lambda}^{-}}\tr(\mathcal{O}Q^{\perp}T_{\sigma})Q^{\perp}}T_{\sigma}^{-1}\chi^{\lambda}(\tau)T_{\tau}\nonumber\\
&&=\frac{1}{(4!)^2}\sum_{\lambda \vdash 4, l(\lambda)\le d} d_{\lambda}^{2}\sum_{\tau\sigma\in\mathcal{S}_4}\parent{\frac{1}{D_{\lambda}^{+}}\tr(\mathcal{O}QT_{\sigma})Q+\frac{1}{D_{\lambda}^{-}}\tr(\mathcal{O}Q^{\perp}T_{\sigma})Q^{\perp}}\chi^{\lambda}(\tau)T_{\sigma^{-1}\tau}\nonumber\\&&=\frac{1}{(4!)^2}\sum_{\lambda \vdash 4, l(\lambda)\le d} d_{\lambda}^{2}\sum_{\pi,\sigma\in\mathcal{S}_{4}}\parent{\frac{1}{D_{\lambda}^{+}}\tr(\mathcal{O}QT_{\sigma})Q+\frac{1}{D_{\lambda}^{-}}\tr(\mathcal{O}Q^{\perp}T_{\sigma})Q^{\perp}}\chi^{\lambda}(\sigma\pi)T_{\pi}
\ea
At this point, we just define
\be
W^{\pm}_{\pi\sigma}=\sum_{\substack{\lambda\vdash 4\\D^{\pm}_{\lambda}\neq 0}}\frac{d_{\lambda}^2}{(4!)^2}\frac{\chi^{\lambda}(\pi\sigma)}{D^{\pm}_\lambda}
\ee
and the derivation is complete. 
\end{proof}

An important property that will be used throughout the paper is the following:
\be
[Q,T_{\pi}]=0, \quad \forall\, \pi\in\mathcal{S}_{4}
\ee
another important property is that $Q$ is a projector, namely $Q^2=Q$.
Another useful result, related to the generalized Weingarten functions is
\be
\sum_{\pi,\sigma\in\mathcal{S}_4}W^{\pm}_{\pi\sigma\in\mathcal{S}_4}T_{\pi}=\frac{\Pi_{\sym}^{(4)}}{D^{\pm}}
\label{weingartenpsym}
\ee
the proof comes from Eq. $\eqref{generalizedWeingarten}$ and from $\Pi_{\sym}^{(4)}=(4!)^{-1}\sum_{\pi\in\mathcal{S}_4}T_{\pi}$. 
 
\subsection{Unitary group average}
Let $\mathcal{O}\in\mathcal{B}(\mathcal{H}^{\otimes t})$ be a bounded operator on $t$-copies of $\mathcal{H}$, then the Haar average reads\cite{collins2003moments,collins2006integration}
\be
\aver{\mathcal{O}_U}_{U\in\haar}=\int_{\mathcal{U}(d)}\de U U^{\dag\otimes t}\mathcal{O}U^{\otimes t}=\sum_{\pi,\sigma\in\mathcal{S}_{t}}W_{\pi\sigma}\tr(\mathcal{O}T_{\sigma})T_{\pi}
\label{haarformulageneral}
\ee
where $T_{\pi}$ is the permutation operator corresponding to the permutation $\pi \in S_{t}$, the $t$-dimensional symmetric group and $W_{\pi\sigma}$ are the Weingarten functions defined as
\be
W_{\pi\sigma}=\sum_{\lambda\vdash t}\frac{d_{\lambda}^{2}}{(t!)^2}\frac{\chi^{\lambda}(\pi\sigma)}{D_{\lambda}}\label{weingartenhaar}
\ee
where $D_{\lambda}=\tr(\Pi^{(4)}_{\lambda})$.
 
 \subsection{A couple of Haar averages over $\mathcal{U}(d)$}
\subsubsection{The average purity\label{puraveragesec}}
Let us calculate the average purity for the output state $\psi_{U}$, for $U\in\mathcal{U}(d)$ or $U\in\mathcal{C}(d)$; indeed the result of the average for the two groups is the same because the Clifford group forms a unitary $3$-design and being a $t$-design means being a $\tilde{t}$-design for any $\tilde{t}\le t$. Then, the average purity
\be
\aver{\pur(\psi_{U})_A}_{U\in\mathcal{U}(d)}=\int_{\haar} \de U\tr(T_{(12)}^{(A)}U^{\otimes 2}\psi^{\otimes 2}U^{\dag\otimes 2})=\sum_{\pi,\sigma\in \mathcal{S}_2}W_{g}(\pi\sigma)\tr(\psi^{\otimes 2}T_{\sigma})\tr(T_{(12)}^{(A)}T_{\pi})
\ee
Since $T_{\sigma}\psi^{\otimes 2}=\psi^{\otimes 2}$ as long as $\psi$ is a pure state, we have
\be
\aver{\pur(\psi_{U})_A}_{U\in\mathcal{U}(d)}=\sum_{\pi,\sigma\in \mathcal{S}_2}W_{g}(\pi\sigma)\tr(T_{(12)}^{(A)}T_{\pi})=\frac{1}{D_\sym}\tr(T_{(12)}^{(A)}\Pi_\sym^{(2)})
\ee
where $\Pi_\sym^{(2)}\equiv\frac{1}{2}\sum_{\sigma\in\mathcal{S}_2}T_{\sigma}$ and $D_{\sym}=\tr(\Pi_{\sym}^{(2)})=d(d+1)/2$. Then, since $T_{\sigma}=T_{\sigma}^{(A)}\otimes T_{\sigma}^{(B)}$, see Sec. \ref{lemma} for a more rigorous treatment, we have
\ba
\aver{\pur(\psi_{U})_A}_{U\in\mathcal{U}(d)}&=&\frac{1}{d(d+1)}\left(\tr(T_{(12)}^{(A)}T_{(12)}^{(A)})\tr(T_{(12)}^{(B)})+\tr(T_{(12)}^{(A)}\bbbone^{(A)\otimes 2})\tr(\bbbone^{(B)\otimes 2})\right)\nonumber\\&=&\frac{2}{d(d+1)}(d_{A}^{2}d_{B}+d_{A}d_{B}^{2})=\frac{d_A+d_B}{d_{A}d_{B}+1}
\ea
where we have used $d=d_Ad_B$.
\subsubsection{The average state $\aver{\psi_{U}^{\otimes 4}}_{U\in\mathcal{U}(d)}$}
Let $\psi_{U}$ be the output state of a quantum circuit $U$. Let us average the fourth tensor power of this output state for $U\in\mathcal{U}(d)$. Using formula \eqref{haarformulageneral} we have
\be
\aver{\psi_{U}^{\otimes 4}}_{U\in\mathcal{U}(d)}= \sum_{\pi,\sigma\in\mathcal{S}_4}W_{\pi\sigma}\tr(T_\sigma \psi^{\otimes 4})T_{\pi}=\sum_{\pi,\sigma\in\mathcal{S}_4}W_{\pi\sigma}T_{\pi}=\frac{\Pi_{\sym}^{(4)}}{D_{\sym}}
\label{haaraverageofstate}
\ee
where we used the fact that $T_\sigma\psi^{\otimes 4}=\psi^{\otimes 4}$ for any permutation operator $T_\sigma$ and that $\sum_{\pi,\sigma\in\mathcal{S}_4}W_{\pi\sigma}T_\pi=\Pi^{(4)}_\sym/D_\sym$, where $D_\sym=\tr(\Pi^{(4)}_\sym)$.

\section{Proofs}
\subsection{Proof of Theorem \ref{th1}}\label{proofmaintheorem}
Let $\mathcal{O}\in\mathcal{B}(\mathcal{H}^{\otimes 4})$  and let the $\Phi^{(4)}_{\mathcal{C}_k} (\mathcal O)$ be its output through the $(\mathcal{C}_k,4)$-fold channel. Then
\be
\Phi^{(4)}_{\mathcal{C}_k}(\mathcal{O}) :=\aver{\mathcal{O}_{U}}_{U\in\mathcal{C}_k}\equiv \langle C_{k}^{\otimes 4}K_{i_k}^{\otimes 4}\dots C_{1}^{\otimes 4}K_{i_{1}}^{\otimes 4}C_{0}^{\otimes 4}\mathcal{O}C_{0}^{\dag\otimes 4}K_{i_1}^{\dag\otimes 4}C_{1}^{\otimes 4}\dots K_{i_{k}}^{\dag\otimes 4}C_{k}^{\dag\otimes 4}\rangle_{C_1\ldots C_k\in\mathcal{C}(d)}
\label{Uk}
\ee
then since the averages over $C_{i}$ for $i=1,\dots, k$ are independent from each other, we can also write
\be
\aver{\mathcal{O}_{U}}_{U\in\mathcal{C}_k}=\aver{\mathcal{C}_{k}^{\otimes 4}K_{i_k}^{\otimes 4} \dots K_{i_{2}}^{\otimes 4}\aver{C_{1}^{\otimes 4}K_{i_{1}}^{\otimes 4}\aver{\mathcal{C}_{0}^{\otimes 4}\mathcal{O}\mathcal{C}_{0}^{\dag\otimes 4}}_{C_0\in\mathcal{C}(d)}K_{i_{1}}^{\dag\otimes 4}C_{1}^{\dag\otimes 4}}_{C_1\in\mathcal{C}(d)}K_{i_{2}}^{\dag\otimes 4}\dots K_{i_k}^{\dag\otimes 4}\mathcal{C}^{\otimes 4}_{k}}_{C_k\in\mathcal{C}(d)}
\ee
The first Clifford average before inserting any single qubit $K$-gate reads
\be
\aver{\mathcal{O}_{U}}_{U\in\mathcal{C}_0}=\sum_{\pi,\sigma\in\mathcal{S}_4}[W^{+}_{\pi\sigma}\tr(\mathcal{O}QT_{\sigma})-W^{-}_{\pi\sigma}\tr(\mathcal{O}Q^{\perp}T_{\sigma})]QT_{\pi}+\sum_{\pi,\sigma\in\mathcal{S}_4}W^{-}_{\pi\sigma}\tr(\mathcal{O}Q^{\perp}T_{\sigma})T_\pi
\ee
where we have used Eq.\eqref{cliffordweing}. We can recast it as
\be
\aver{\mathcal{O}_{U}}_{U\in\mathcal{C}_0}=\sum_{\pi\in\mathcal{S}_4}(c_{\pi}(\mathcal{O})Q+b_{\pi}(\mathcal{O}))T_{\pi}
\ee
where
\ba
c_{\pi}(\mathcal{O})&=&\sum_{\sigma\in\mathcal{S}_4}[W^{+}_{\pi\sigma}\tr(\mathcal{O}QT_{\sigma})-W^{-}_{\pi\sigma}\tr(\mathcal{O}Q^{\perp}T_{\sigma})]\\
b_{\pi}(\mathcal{O})&=&\sum_{\sigma\in\mathcal{S}_4}W^{-}_{\pi\sigma}\tr(\mathcal{O}Q^{\perp}T_{\sigma})
\ea
Now we need to apply the first $K_{i_1}$-gate on the $i$-th qubit; noting that $[T_{\pi},K_{i_1}^{\otimes 4}]=0
$ for all $\pi\in\mathcal{S}_4$, we have
\be
\aver{\mathcal{O}_{U}}_{U\in\mathcal{C}_0}\rightarrow K^{\dag\otimes 4}_{i_1}\aver{\mathcal{O}_{U}}_{U\in\mathcal{C}_0}K_{i_1}^{\otimes 4}=\sum_{\pi\in\mathcal{S}_4}(c_{\pi}(\mathcal{O})K_{i_1}^{\otimes 4}QK_{i_1}^{\dag\otimes 4}+b_{\pi}(\mathcal{O}))T_{\pi}
\ee
and then average over another Clifford layer, knowing that the Clifford operator only acts non trivially only on the operator  $K_{i_1}^{\otimes 4}QK_{i_1}^{\dag\otimes 4}T_{\pi}$ because $[C^{\otimes 4},T_\sigma]=0,\,\forall\sigma$
\be
\aver{\mathcal{O}_{U}}_{U\in\mathcal{C}_1}=\aver{C_{1}^{\otimes 4}K_{i_1}^{\otimes 4}\aver{\mathcal{O}_{U}}_{U\in\mathcal{C}_0}K_{i_1}^{\dag\otimes 4}C_{1}^{\dag\otimes 4}}_{C_1\in \mathcal{C}(d)}=\sum_{\pi\in\mathcal{S}_4}c_\pi(\mathcal{O})\aver{C_{1}^{\dag\otimes 4}K_{i_1}^{\otimes 4}QK_{i_1}^{\dag\otimes 4}T_\pi C_{1}^{\dag\otimes 4}}_{C_1\in\mathcal{C}(d)}+b_{\pi}(\mathcal{O})T_\pi
\ee
Then, from Eq.\eqref{cliffordweing}
\be
\aver{C_{1}^{\dag\otimes 4}K_{i_1}^{\otimes 4}QK_{i_1}^{\dag\otimes 4}T_\pi C_{1}^{\dag\otimes 4}}_{C_1\in\mathcal{C}(d)}=\sum_{\sigma\in\mathcal{S}_4}(\Xi_{\pi\sigma}Q+\Lambda_{\pi\sigma})T_\sigma
\ee
where $\Xi_{\pi\sigma}$ and $\Lambda_{\pi\sigma}$ read
\ba
\Xi_{\sigma\pi}&\equiv&\sum_{\tau\in\mathcal{S}_4}[W^{+}_{\pi\tau}\tr(T_{\sigma}K^{\otimes 4}QK^{\dag\otimes 4}QT_{\tau})-W^{-}_{\pi\tau}\tr(T_{\sigma}K^{\otimes 4}QK^{\dag\otimes 4}Q^{\perp}T_{\tau})]\\
\Lambda_{\sigma\pi}&\equiv&\sum_{\tau\in\mathcal{S}_4}W^{-}_{\pi\tau}\tr(T_{\sigma}K^{\otimes 4}QK^{\dag\otimes 4}Q^{\perp}T_{\tau})
\ea
we have defined the matrix $\Xi$ omitting the subscript $K_{i_{1}}$ because, as shown in Lemma \ref{lemma1}, it does not play any role. Thus, we have $\aver{\mathcal{O}_{U}}_{U\in\mathcal{C}_1}=\sum_{\pi,\sigma\in\mathcal{S}_4}(\Xi_{\sigma\pi}Q+\Lambda_{\sigma\pi})c_{\pi}(\mathcal{O})T_\sigma+b_{\pi}(\mathcal{O})T_{\pi}$; at the next iteration
\be
\aver{\mathcal{O}_{U}}_{U\in\mathcal{C}_2}=\sum_{\pi,\sigma\tau\in\mathcal{S}_4}(\Xi_{\tau\sigma}Q+\Lambda_{\tau\sigma})\Xi_{\sigma\pi}c_{\pi}(\mathcal{O})T_{\tau}+\sum_{\pi,\sigma\in\mathcal{S}_4}[\Lambda_{\sigma\pi}c_{\pi}(\mathcal{O})+\delta_{\pi\sigma}b_{\pi}(\mathcal{O})]T_{\sigma}
\ee
we can recast it as
\be
\aver{\mathcal{O}_{U}}_{U\in\mathcal{C}_2}=\sum_{\pi,\sigma\in\mathcal{S}_4}\left[\parent{(\Xi^2)_{\pi\sigma}Q+\Lambda_{\pi\sigma}+\sum_{\tau\in\mathcal{S}_4}\Lambda_{\pi\tau}\Xi_{\tau\sigma}}c_{\pi}(\mathcal{O})+\delta_{\pi\sigma}b_{\pi}(\mathcal{O})\right]T_{\sigma}
\ee
The latter relationship can be easily generalized to $k$ layers as
\be
\aver{\mathcal{O}_{U}}_{U\in\mathcal{C}_k}=\sum_{\pi,\sigma\in\mathcal{S}_4}\left[\parent{(\Xi^{k})_{\pi\sigma}Q+\Gamma_{\pi\sigma}^{(k)}}c_{\pi}(\mathcal{O})+\delta_{\pi\sigma}b_{\pi}(\mathcal{O})\right]T_{\sigma}
\ee
where we have defined $\Gamma^{(k)}_{\pi\sigma}\equiv\sum_{\tau\in\mathcal{S}_4}\Lambda_{\pi\tau}\sum_{i=0}^{k-1}(\Xi^i)_{\tau\sigma}$. This concludes the proof. \qed

\subsection{Proof of Application \ref{thconvergence}\label{proofthconvergence}}
From theorem \ref{th1}, the $(\mathcal{C}_k,4)$-fold channel reads
\be
\Phi^{(4)}_{\mathcal{C}_k}(\mathcal{O})=\sum_{\pi,\sigma\in\mathcal{S}_{4}}\left[\parent{(\Xi^{k})_{\pi\sigma}Q+\Gamma_{\pi\sigma}^{(k)}}c_{\pi}(\mathcal{O})+\delta_{\pi\sigma}b_{\pi}(\mathcal{O})\right]T_{\sigma}
\label{76}
\ee
First of all let us write this equation in matrix form for the coefficients; define $\mathbf{T}$ a vector with components the permutation operators $T_{\sigma}$, $\mathbf{c}$ the vector with components $c_{\pi}(\mathcal{O})$ and similarly for $\mathbf{b}$, then Eq.\eqref{76} becomes
\be
\Phi^{(4)}_{\mathcal{C}_k}(\mathcal{O})=(\Xi^k\cdot \mathbf{c},\mathbf{T})Q+(\Gamma^{(k)}\cdot \mathbf{c}+\mathbf{b},\mathbf{T})
\label{matrixformouk}
\ee
where the $\cdot$ stands for the row by column product and $(\cdot,\cdot)$ for the usual scalar product between lists. Recall that for the Unitary group the $(\mathcal{U}(d),4)$-fold channel reads
\be
\Phi^{(4)}_{\mathcal{U}(d)}(\mathcal{O})=(W\cdot \mathbf{t},\mathbf{T})
\ee
where $W$ is the matrix with components the Unitary group Weingarten functions, cfr Eq.\eqref{weingartenhaar}.
In the following we prove that the first piece in Eq.\eqref{matrixformouk} vanishes in the limit $k\rightarrow\infty$, while the second returns the matrix $W$.
\begin{lemma} For $K=P_{\theta}\equiv \ket{0}\bra{0}+e^{i\theta}\ket{1}\bra{1}$, the matrix $\Xi$, defined in Eq.\eqref{xi} has the following properties
\begin{itemize}
    \item $\Xi$ is symmetric;
    \item $\Xi$ has rank 6;
    \item the eigenvalues read
    \ba
    \lambda^{(\pm)}&=&f^{\pm}_{\theta}, \quad \mu(\lambda^{\pm})=1\nonumber\\
    \lambda^{av}&=&\frac{f^{+}_{\theta}+f^{-}_{\theta}}{2}, \quad \mu(\lambda^{av})=4
    \ea
    where $\mu(\lambda)$ stands for the algebraic multiplicity of the eigenvalue $\lambda$ and \ba\label{fpmtheta}
    f^{\pm}_{\theta}&=&\frac{7d^2\pm 3d +d(d\mp 3)\cos(4\theta)-8}{8(d^2-1)}< 1, \quad \forall\,\theta\neq \pm\frac{\pi}{2}\\
    f_{\theta}^{\pm}&=&1, \quad \theta=\pm\frac{\pi}{2}\nonumber
    \ea

    \item the maximum eigenvalue is $f^{+}_{\theta}$.
\end{itemize}

\end{lemma}
The proof comes from direct calculation of the $24\times 24$ matrix $\Xi$ with $K=P_{\theta}$. \qed

Since all the eigenvalues of $\Xi$ are less than $1$,
\be
\lim_{k\rightarrow\infty}(\Xi^k)_{\pi\sigma}=0, \quad \text{for all $\pi,\sigma$}
\ee
and
\be
\sum_{i=0}^{\infty}\Xi^i=(1-\Xi)^{-1}
\ee
thus $\Gamma^{(\infty)}=\Lambda (1-\Xi)^{-1}$.
Defining the vector $\mathbf{q}$ having components $\tr(\mathcal{O}QT_\sigma)$ and the vector $\mathbf{t}$ having components $\tr(\mathcal{O}T_{\sigma})$, from Eq.\eqref{ccoeff} and Eq.\eqref{bcoeff} we note that
\ba
\mathbf{c}&=&(W^++W^-)\cdot\mathbf{q}-W^- \cdot\mathbf{t}\nonumber\\
\mathbf{b}&=&W^{-}\cdot \mathbf{t}-W^{-}\cdot \mathbf{q}
\ea
where $W^{\pm}$ are the matrices with components the generalized Weingarten functions for the Clifford group, cfr. \eqref{generalizedWeingarten}. Therefore taking the limit $k\rightarrow\infty$ in Eq.\eqref{matrixformouk}
\be
\lim_{k\rightarrow\infty}\Phi^{(4)}_{\mathcal{C},k}(\mathcal{O})=\left([\Lambda (1-\Xi)^{-1}(W^{+}+W^{-})-W^-]\cdot \mathbf{q},\mathbf{T}\right)+\parent{[W^{-}-\Lambda(1-\Xi)^{-1}W^{-}]\cdot \mathbf{t},\mathbf{T}}
\ee
It is straightforward to check that
\be
W^{-}-\Lambda(1-\Xi)^{-1}W^{-}=W
\ee
Then, the vector $\mathbf{q}$ is in kernel of the matrix $\zeta=\Lambda (1-\Xi)^{-1}(W^{+}+W^{-})-W^-$, namely $\kerf(\zeta)$; indeed
\be
\kerf(\zeta)=\spann{\mathbf{e}^{(\alpha)}\,|\,\alpha=1,\dots,6}
\ee
where the non null components of these six vectors read
\ba
e_{(e)}^{(1)}&=&e_{(12)(34)}^{(1)}=e_{(13)(24)}^{(1)}=e_{(14)(23)}^{(1)}\nonumber\\
e_{(13)}^{(2)}&=&e_{(24)}^{(2)}=e_{(1432)}^{(2)}=e_{(1234)}^{(2)}\nonumber\\
e_{(14)}^{(3)}&=&e_{(23)}^{(3)}=e_{(1342)}^{(3)}=e_{(1243)}^{(3)}\nonumber\\
e_{(12)}^{(4)}&=&e_{(34)}^{(4)}=e_{(1324)}^{(4)}=e_{(1423)}^{(4)}\nonumber\\
e_{(132)}^{(5)}&=&e_{(124)}^{(5)}=e_{(143)}^{(5)}=e_{(234)}^{(5)}\nonumber\\
e_{(123)}^{(6)}&=&e_{(142)}^{(6)}=e_{(134)}^{(6)}=e_{(243)}^{(6)}
\ea
because of Lemma \ref{lemmaPauli} it is clear that $\mathbf{q}\in\kerf(\zeta)$, which proves the theorem. \qed


\subsection{Proof of Application \ref{application1}\label{proofcorollary1}}
We will make use of Theorem \ref{th1}
\be
\Phi^{(4)}_{\mathcal{C}_k}(\psi^{\otimes 4})=\sum_{\pi,\sigma\in\mathcal{S}_{4}}\left[\parent{(\Xi^{k})_{\pi\sigma}Q+\Gamma_{\pi\sigma}^{(k)}}c_{\pi}(\psi^{\otimes 4})+\delta_{\pi\sigma}b_{\pi}(\psi^{\otimes 4})\right]T_{\sigma}
\ee
Let us first compute the term $\sum_{\pi,\sigma\in\mathcal{S}_4}c_{\sigma}(\psi^{\otimes 4})(\Xi^k)_{\pi\sigma} Q T_\pi $; note that $[T_{\pi},Q]=0$. This is a fact that will be repeatedly exploited in this proof. First we prove that $c_{\sigma}=c$ independent from the specific permutation $\sigma$; from Eq.\eqref{bc_coeff} we have
\ba
c_{\sigma}(\psi^{\otimes 4})&=&\sum_{\pi\in\mathcal{S}_4}[W^{+}_{\pi\sigma}\tr(\psi^{\otimes 4}QT_{\pi})-W^{-}_{\pi\sigma}\tr(\psi^{\otimes 4}Q^{\perp}T_{\pi})]=\sum_{\pi\in\mathcal{S}_4}W^{+}_{\pi\sigma}\tr(Q\psi^{\otimes 4})-W^{-}_{\pi\sigma}\tr(Q^\perp\psi^{\otimes 4})\nonumber\\&=&\frac{1}{4!D^+}\tr(Q\psi^{\otimes 4})-\frac{1}{4!D^-}\tr(Q^\perp \psi^{\otimes 4})\equiv c
\ea
where we used $T_{\pi}\psi^{\otimes 4}=\psi^{\otimes 4}$ for all $\pi$ and $\sum_{\pi\in\mathcal{S}_4}W^{\pm}_{\pi\sigma}=(4!D^{\pm})^{-1}$. Then the sum can be written as
\ba
\sum_{\pi,\sigma\in\mathcal{S}_4}c_{\pi}(\psi^{\otimes 4})(\Xi^k)_{\pi\sigma} Q T_\sigma&=&c\sum_{\pi,\sigma\in\mathcal{S}_4}(\Xi^k)_{\pi\sigma} Q T_\sigma
=c\sum_{\pi,\sigma,\tau\in\mathcal{S}_4}\Xi_{\pi\tau}(\Xi^{k-1})_{\tau\sigma}QT_{\sigma}\nonumber\\&=&c\sum_{\sigma,\tau\in\mathcal{S}_4}\left(\sum_{\pi}\Xi_{\pi\tau}\right)(\Xi^{k-1})_{\tau\sigma}T_{\sigma}Q\label{xidec}
\ea
let us prove that $\sum_{\pi\in\mathcal{S}_4}\Xi_{\pi\tau}$ does not depend on $\tau$; from Eq.\eqref{xi} it is easy to see that
\be
\sum_{\pi\in\mathcal{S}_4}\Xi_{\pi\tau}=\frac{1}{D^{+}}\tr(\Pi_{\sym}^{(4)}K^{\otimes 4}QK^{\dag\otimes 4}Q)-\frac{1}{D^-}\tr(\Pi_{\sym}^{(4)}K^{\otimes 4}QK^{\dag\otimes 4}Q^{\perp})=\frac{c_Q}{D^+}-\frac{c_{QQ^{\perp}}}{D^-}
\ee
where we have used $\sum_{\pi\in\mathcal{S}_4}W^{\pm}_{\pi\sigma}=(4!D^{\pm})^{-1}$ and $\Pi_{\sym}^{(4)}T_{\tau}=\Pi_\sym^{(4)}$ for any $\tau\in\mathcal{S}_4$. The coefficients $c_Q$ and $c_{QQ^\perp}$ are to be computed in a straightforward way; for the case $K=P_{\theta}$ they are explicitly calculated in App. \ref{calculationcqcqqp}. Since $\sum_{\pi\in\mathcal{S}_4} \Xi_{\pi\tau}$ does not depend on $\tau$, the decomposition introduced in the last equality of Eq. $\eqref{xidec}$ can be reiterated $k$ times to  obtain 
\ba
\sum_{\pi,\sigma\in\mathcal{S}_4}c_{\pi}(\psi^{\otimes 4})(\Xi^k)_{\pi\sigma} Q T_\sigma&=&c\parent{\frac{c_Q}{D^+}-\frac{c_{QQ^\perp}}{D^-}}^{k}\sum_{\sigma\in\mathcal{S}_4}T_{\sigma}Q\nonumber\\&=&\parent{\frac{\tr(Q\psi^{\otimes 4})}{D^+}-\frac{\tr(\psi^{\otimes 4}Q^{\perp})}{D^-}}\parent{\frac{c_Q}{D^+}-\frac{c_{QQ^\perp}}{D^-}}^{k}\Pi_{\sym}^{(4)}Q
\label{s100}
\ea
Now we compute the term $\sum_{\pi,\sigma\in\mathcal{S}_4}\Gamma^{(k)}_{\pi\sigma}c_{\pi}(\psi^{\otimes 4})T_{\sigma}$ where $\Gamma^{(k)}_{\pi\sigma}$ is defined in Theorem \ref{th1}; since $c_{\pi}(\psi^{\otimes 4})$ is independent from the permutation $\pi$ we can write
\be 
\sum_{\pi,\sigma\in\mathcal{S}_4}\Gamma^{(k)}_{\pi\sigma}c_{\pi}(\psi^{\otimes 4})T_{\sigma}=c\sum_{\pi,\sigma\in\mathcal{S}_4}\Gamma^{(k)}_{\pi\sigma}T_{\sigma}=c\sum_{\pi,\sigma,\tau\in\mathcal{S}_4}\Lambda_{\pi\tau}\left(\sum_{i=0}^{k-1}(\Xi^{i})_{\tau\sigma}\right)T_{\sigma}
\ee 
It is easy to see $\sum_{\pi\in\mathcal{S}_4}\Lambda_{\pi\tau}=c_{QQ^\perp}/D^{-}$ does not depend on $\tau$; from this fact, we can use the same technique used above to compute $\sum_{\tau\in\mathcal{S}_4}\left(\sum_{i=0}^{k-1}(\Xi^{i})_{\tau\sigma}\right)$ and finally obtain
\be 
\sum_{\pi,\sigma\in\mathcal{S}_4}\Gamma^{(k)}_{\pi\sigma}c_{\pi}(\psi^{\otimes 4})T_{\sigma}=\frac{\tr(\psi^{\otimes 4}Q^{\perp})}{D^{-}}+\frac{c_{QQ^\perp}}{D^-}\left(\frac{\tr(Q\psi^{\otimes 4})}{D^+}-\frac{\tr(Q^\perp \psi^{\otimes 4})}{D^-}\right) \sum_{i=0}^{k-1}\left(\frac{c_Q}{D^+}-\frac{c_{QQ^{\perp}}}{D^-}\right)^{i} \Pi_{\sym}^{(4)}
\label{s101}
\ee 
The last term we need to evaluate is $\sum_{\pi}b_{\pi}(\psi^{\otimes 4})T_{\pi}$; as before, let us prove that $b_{\pi}(\psi^{\otimes 4})$ does not depend on $\pi$
\be
b_{\pi}(\psi^{\otimes 4})=\sum_{\sigma\in\mathcal{S}_4} W_{\pi\sigma}^{-}\tr\parent{\psi^{\otimes 4} Q^{\perp} T_{\sigma}}=\sum_{\sigma\in\mathcal{S}_4}W_{\pi\sigma}^{-}\tr\parent{\psi^{\otimes 4}Q^{\perp}} = \frac{1}{4!D^{-}}\tr({Q^{\perp}\psi^{\otimes 4}})
\ee
then
\be
\sum_{\pi\in\mathcal{S}_4}b_{\pi}(\psi^{\otimes 4})T_{\pi}=\frac{1}{D^{-}}\tr(Q^{\perp}\psi^{\otimes 4})\frac{1}{4!}\sum_{\pi\in\mathcal{S}_4}T_{\pi}=\frac{1}{D^{-}}\tr(Q^{\perp}\psi^{\otimes 4})\Pi_{\sym}^{(4)}
\label{s207}
\ee
Putting together Eqs. \eqref{s100}, \eqref{s101} and \eqref{s207} we obtain the final result in Eq.\eqref{maintheorem}. \qed
\subsection{Proof of Application \ref{app4}\label{calculationcqcqqp}}
In this section we evaluate $c_{Q}=\tr(K_{i_{j}}^{\otimes 4}Q K_{i_{j}}^{\dag\otimes 4}Q\Pi_{\sym}^{(4)})$ for $K= P_{\theta}\equiv \ket{0}\bra{0}+e^{i\theta}\ket{1}\bra{1}$. As pointed out in Lemma \ref{lemma1}, the position of the operator $K_{i_j}$ does not affect the calculations, so in the following we analyze the case in which the operator $K_{i}$ acts on the first qubit $i_{1}$.
The term $c_{Q}$ can be rewritten as
\be
\tr(K_{i_{1}}^{\otimes 4}QK_{i_{1}}^{\dag\otimes 4}Q\Pi_{\sym}^{(4)})=\tr(R^{-1}RK^{\otimes 4}_{i_{1}}R^{-1}RQR^{-1}RK_{i_{1}}^{\dag\otimes 4}R^{-1}RQR^{-1}R\Pi_{\sym}^{(4)})
\ee 
where $R\in \mathcal{S}_{4N}$ is a permutation operator whose action on a tensor product basis element is
\be
R\ket{i_{1} \ldots i_{N}}^{\otimes 4}\equiv\ket{i_{1}}^{\otimes 4}\ket{i_{2}\ldots i_{N}}^{\otimes 4}
\ee 
The adjoint action of the permutation $R$ allows us to rewrite the operator $Q$ as 
\be 
RQR^{-1}=\frac{1}{4}(Q_I+Q_X+Q_Y+Q_Z)
\ee 
where $X,Y,Z,I$ are single qubit Pauli matrices and for example $Q_X$ reads
\be
Q_{X}=\frac{1}{(d/2)^2}\sum_{P\in\mathcal{P}(d/2)}X^{\otimes 4}\otimes P^{\otimes 4}=X^{\otimes 4}\otimes Q_{d/2}
\ee
and $Q_{d/2}=(d/2)^{-2}\sum_{P\in\mathcal{P}(d/2)}P^{\otimes 4}$.
Similarly the adjoint action of $R$ on $K^{\otimes 4}$
\be
RK^{\otimes 4}_{i_1}R^{-1}=K^{\otimes 4}\otimes (\bbbone_{2}\otimes\dots\otimes \bbbone_{N})^{\otimes 4}
\ee
acting on the first qubit element of each copy of $\mathcal{H}$ it acts on the first four elements of $RQR^{-1}$. It is simple to see that the adjoint action of $R$ on the symmetric projector $\Pi_{\sym}^{(4)}$ give us again the symmetric projector on the new permuted space and we denote it with $\tilde{\Pi}_{\sym}^{(4)}$.
The term $c_{Q}$ can be rewritten as
\ba
\tr(K_{i_{1}}^{\otimes 4}QK_{i_{1}}^{\dag\otimes 4}Q\tilde{\Pi}_{\sym}^{(4)})&=&\frac{1}{16}\tr(RK^{\otimes 4}_{i_1}R^{-1}(Q_I+Q_X+Q_Y+Q_Z)RK^{\dag\otimes 4}_{i_1}R^{-1}(Q_I+Q_X+Q_Y+Q_Z)\tilde{\Pi}_{\sym}^{(4)})\nonumber\\&=&\frac{1}{16}\tr[(Q_I+Q_Z)(Q_I+Q_X+Q_Y+Q_Z)\tilde{\Pi}_{\sym}^{(4)}]\nonumber\\&+&\tr[K_{i_{1}}^{\otimes 4}(Q_X+Q_Y)K_{i_{1}}^{\dag\otimes 4}(Q_I+Q_X+Q_Y+Q_Z)\tilde{\Pi}_{\sym}^{(4)}]\label{kacqdec}
\ea
where we used the fact that $[K,I]=[K,Z]=0$; then the first term of Eq.$\eqref{kacqdec}$ reads
\be
\tr[(Q_I+Q_Z)(Q_I+Q_X+Q_Y+Q_Z)\tilde{\Pi}_{\sym}^{(4)}]=8\tr(Q\tilde{\Pi}_{\sym}^{(4)})=8D^{+}.
\ee
We focus now on the second term of Eq.$\eqref{kacqdec}$
\ba
\tr[K_{i_{1}}^{\otimes 4}(Q_X+Q_Y)K_{i_{1}}^{\dag\otimes 4}(Q_I+Q_X+Q_Y+Q_Z)\tilde{\Pi}_{\sym}^{(4)}]&=&\tr[(Q_{KXK^{\dag}}+Q_{KYK^{\dag}})(Q_I+Q_X+Q_Y+Q_Z)\tilde{\Pi}_{\sym}^{(4)}]\nonumber\\
\nonumber&=&2\tr[(Q_{KXK^{\dag}}+Q_{KXK^{\dag}X}+Q_{KXK^{\dag}Y}+Q_{KYK^{\dag}})\tilde{\Pi}_{\sym}^{(4)}]\\
&=&2\tr(\tilde{Q}\tilde{\Pi}_{\sym}^{(4)})=\frac{1}{12}\sum_{\sigma\in\mathcal{S}_4}\tr(\tilde{Q}\tilde{T}_{\sigma})
\ea
where denoted $\tilde{T}_{\sigma}=RT_{\sigma}R^{-1}$ and defined
\be
\tilde{Q}\equiv(Q_{KXK^{\dag}}+Q_{KXK^{\dag}X}+Q_{KXK^{\dag}Y}+Q_{KYK^{\dag}})
\ee
and
\ba
Q_{KXK^{\dag}X}&=&(KXK^{\dag}X)^{\otimes 4}\otimes Q_{d/2}, \\
Q_{KXK^{\dag}Y}&=&(KXK^{\dag}Y)^{\otimes 4}\otimes Q_{d/2}, \nonumber\\
Q_{KXK^{\dag}}&=&(KXK^{\dag})^{\otimes 4}\otimes Q_{d/2}\nonumber \\
Q_{KYK^{\dag}}&=&(KYK^{\dag})^{\otimes 4}\otimes Q_{d/2}\nonumber
\ea
Therefore
\be
c_{Q}=\tr(K^{\otimes 4}QK^{\dag\otimes 4}Q\Pi_{\sym}^{(4)})=\frac{D^{+}}{2}+\frac{\sum_{\sigma\in\mathcal{S}_4}\tr(\tilde{Q}\tilde{T}_{\sigma})}{192}
\ee
In a similar fashion  to what we have done in Sec. \ref{lemma} it is possible to see that $\tilde{T}_{\sigma}=T_{\sigma}^{(2)}\otimes T_{\sigma}^{(d/2)}$ where $T_{\sigma}^{(2)}\in \mathcal{B}((\mathbb{C}^2)^{\otimes 4})$ and $T_{\sigma}^{(d/2)}\in \mathcal{B}((\mathbb{C}^{2\otimes (N-1)})^{\otimes 4})$; thus, the following equality holds
\ba
&&\tr[(Q_{KXK^{\dag}}+Q_{KXK^{\dag}X}+Q_{KXK^{\dag}Y}+Q_{KYK^{\dag}})\tilde{\Pi}_{\sym}^{(4)}]=\\
&&\frac{1}{12}\sum_{\sigma\in\mathcal{S}_4}\tr[T_{\sigma }^{(2)}((KXK^{\dag})^{\otimes 4}+(KYK^{\dag})^{\otimes 4}-(KXK^{\dag}X)^{\otimes 4}-(KXYK^{\dag}Y)^{\otimes 4})]\tr(Q_{d/2}T_{\sigma}^{(d/2)})\nonumber
\ea 
It is easy to see that
\ba
(KXK^{\dag})^{\otimes 4}=\begin{pmatrix}0&e^{-i\theta}\\
e^{i\theta}&0 
\end{pmatrix}^{\otimes 4},&& \quad (KYK^{\dag})^{\otimes 4}=\begin{pmatrix}0&-ie^{-i\theta}\\
ie^{i\theta}&0
\end{pmatrix}^{\otimes 4},\\ (KXK^{\dag}X)^{\otimes 4}=\begin{pmatrix}e^{-i\theta}&0\\
0& e^{i\theta}
\end{pmatrix}^{\otimes 4}, && \quad
(KXK^{\dag}Y)^{\otimes 4}=\begin{pmatrix}
ie^{-i\theta}&0\\0&-ie^{i\theta}
\end{pmatrix}^{\otimes 4},\nonumber
\ea
With all the previous consideration we are now ready to compute the coefficients
\ba
1\times\tr(\tilde{Q})&=&16(\cos^4\theta+\sin^4\theta)\tr(Q_{d/2})=4(\cos^4\theta+\sin^4\theta)d^2\nonumber\\
6\times\tr(T_{(ij)}\tilde{Q})&=&8\cos^{2}(2\theta)\tr(T_{(ij)}Q_{d/2})=4d\cos^2(2\theta)\nonumber\\
8\times\tr(T_{(ijk)}\tilde{Q})&=&4\cos(4\theta)\\
6\times\tr(T_{ijkl}\tilde{Q})&=&8\cos^2(2\theta)\tr(T_{(ijkl)}Q_{d/2})=4d\cos^2(2\theta)\nonumber\\
3\times\tr(T_{(ij)(kl)}\tilde{Q})&=&4(3+\cos(4\theta))\tr(T_{(ij)(kl)}Q_{d/2})=d^2(3+\cos(4\theta))\nonumber
\ea
Thus $\frac{1}{12}\sum_{\sigma\in\mathcal{S}_4}\tr(\tilde{Q}\tilde{T}_{\sigma})=\frac{4(d+2)(3d+(d+4)\cos(4\theta))}{12}$ and
\ba
c_{Q}&=&\tr(K_{i_{1}}^{\otimes 4}QK_{i_{1}}^{\dag\otimes 4}Q\Pi_{\sym}^{(4)})=\frac{d^2+3d+2}{12}+\frac{(d+2)(3d+(d+4)\cos(4\theta))}{48}\\
\nonumber &=&\frac{(d+2)((d+4)\cos(4\theta)+7d+4)}{48}\nonumber\\
c_{QQ^{\perp}}&=&\tr(K_{i_{1}}^{\otimes 4}QK_{i_{1}}^{\dag\otimes 4}Q^{\perp}\Pi_{\sym}^{(4)})=\tr(Q\tilde{\Pi}_{\sym}^{(4)})-c_{Q}=\frac{(d+4)(d+2)}{24}\sin^2(2\theta)
\ea
Therefore
\ba
a_{k}&=&\frac{24 (f_{\theta}^{-})^k}{(d^2-1)(d+2)(d+4)}\left(\frac{d(d+3)}{4}\tr(\psi^{\otimes 4} Q)-1\right)\\
b_{k}&=&\frac{1}{D_{\sym}}+\frac{24}{(d^2-1)(d+2)(d+4)}(f_{\theta}^{-})^k\left(\frac{4}{d(d+3)}-\tr(\psi^{\otimes 4}Q)\right)
\ea
where $f_{\theta}^{-}$ is, cfr. \eqref{fpmtheta}
\be
f_{\theta}^{-}\equiv\frac{7d^2-3d+d(d+3)\cos(4\theta)-8}{8(d^2-1)}= \frac{7+\cos(4\theta)}{8}+\Theta(d^{-1})
\ee
It is possible to calculate the extreme points of $f_{\theta}$. the maximum is $f_{\theta}=1$ for $\theta=\pi/2$, while the minimum is $f_{\theta}\approx \frac{3}{4}$ for $\theta=\pi/4$.

\subsection{Proof of Lemma \ref{lemmafluc}}\label{prooflemmafluc}
The average square purity for the $k$-doped Clifford circuit can be written as\cite{Hammalungo_2012}
\be
\aver{\pur^2(\psi_{U})_A}_{U\in\mathcal{C}_k}=\tr\left(T^{(A)}_{(12)(34)}\aver{\psi^{\otimes 4}_{U}}_{U\in\mathcal{C}_k}\right)
\label{purfluctdoped}
\ee
Then, substituting Eq.\eqref{maintheorem}
\be
\aver{\pur^2(\psi_{U})_A}_{U\in\mathcal{C}_k}=a_{k}\tr(Q\Pi_{\sym}^{(4)}T_{(12)(34)}^{(A)})+b_{k}\tr(\Pi_{\sym}^{(4)}T_{(12)(34)}^{(A)})
\ee
where $a_k$ and $b_k$ are defined in Eq.\eqref{coefficientmaintheorem}. Recalling that $Q=Q_{A}\otimes Q_{B}$ up to a rearrangement of the tensor product (cfr. App. \ref{lemma}), then $T_{\sigma}=T^{(A)}_{\sigma}\otimes T_{\sigma}^{(B)}$ and $\Pi_{\sym}^{(4)}=1/24\sum_{\sigma\in\mathcal{S}_4}T_\sigma$, we write
\ba
\tr(Q\Pi_{\sym}^{(4)}T_{(12)(34)}^{(A)})&=&\frac{1}{24}\sum_{\sigma\in\mathcal{S}_4}\tr(Q_{A}T_{\sigma}^{(A)}T_{(12)(34)}^{(A)})\tr(Q_{B}T_{\sigma}^{(B)})\\
\tr(\Pi_{\sym}^{(4)}T_{(12)(34)}^{(A)})&=&\frac{1}{24}\sum_{\sigma\in\mathcal{S}_4}\tr(T_{\sigma}^{(A)}T_{(12)(34)}^{(A)})\tr(T_{\sigma}^{(B)})
\ea
after some long but trivial algebra one gets, for the case $d_A=d_B=\sqrt{d}$ and $\psi=\ket{0}\bra{0}^{\otimes N}$
\ba\label{purfluck}
\aver{\pur^2(\psi_{U})_A}_{U\in\mathcal{C}_k}&=&\frac{2(2d^2+9d+1)+(d^2-2d+1)(f^{-}_{\theta})^{k}}{(d+1)(d+2)(d+3)}\\
\aver{\pur^2(\psi_{U})_A}_{U\in\mathcal{U}(d)}&=&\frac{4 d^2+18 d+2}{(d+1) (d+2) (d+3)}
\ea 
and thus, by computing $\aver{\pur^2(\psi_{U})_A}_{U\in\mathcal{C}_k}-\aver{\pur(\psi_U)_A}_{U\in\haar}^{2}$ one  finds Eq.\eqref{fluctclifforddoped}. \qed

\section{Other Proofs}
\subsection{Notes on Pauli operators}
The Pauli operators on $\mathbb{C}^{2\otimes N}$ are formed by all Pauli strings
\be
P=p_{1}\otimes p_{2}\otimes\cdots\otimes p_N
\ee
where $p_{i}\in\{I,X,Y,Z\}$ are usual Pauli matrices on $\mathbb{C}^2$. They are unitary and hermitian operators. Moreover, the Pauli group forms a $1$-design, i.e the $(\mathcal{P}(d),1)$-fold channel of the Pauli group equals the $(\mathcal{U}(d),1)$-fold channel 
\be
\Phi^{(1)}_{\mathcal{P}(d)}(\mathcal{O})\equiv\frac{1}{d^2}\sum_{P\in\mathcal{P}(d)} P\mathcal{O}P=\frac{\tr(\mathcal{O})}{d}
\ee
Since Pauli operators commute or anticommute, we define for $P_{1},P_{2}\in\mathcal{P}(2^N)$
\be
P_{2}P_{1}P_{2}=K(P_{1},P_{2})P_{1} \quad K(P_{1},P_{2}):=\frac{1}{d}\tr(P_{2}^{\dag}P_{1}P_{2}P_{1})
\ee
where $K(P_{1},P_{2})$ is either $1$ or $-1$; a useful rule for combining them
\be
K(P_{1},P_{2})K(P_{1},P_{3})=\frac{1}{d}K(P_{1},P_{2}P_{3})=\frac{1}{d^2}\tr((P_{2}P_{3})^{\dag}P_{1}P_{2}P_{3}P_{1})
\ee
where $P_{1},P_2,P_{3}\in\mathcal{P}(2^N)$. The above facts are sufficient to prove the following lemma
\begin{lemma}\label{lemmaPauli} 
Let $\mathcal{O}\in\mathcal{B}(\mathcal{H}^{\otimes 4})$ and let $Q=d^{-2}\sum_{P\in\mathcal{P}(2^N)}P^{\otimes 4}$, then the following relations hold
\ba
\tr(\mathcal{O}Q)&=&\tr(\mathcal{O}QT_{(12)(34)})=\tr(\mathcal{O}QT_{(13)(24)})=\tr(\mathcal{O}QT_{(14)(23)})\label{154}\\
\tr(\mathcal{O}QT_{(24)})&=&\tr(\mathcal{O}QT_{(1432)})=\tr(\mathcal{O}QT_{(13)})=\tr(\mathcal{O}QT_{(1234)})\\
\tr(\mathcal{O}QT_{14})&=&\tr(\mathcal{O}QT_{(23)})=\tr(\mathcal{O}QT_{(1342)})=\tr(\mathcal{O}QT_{(1243)})\\
\tr(\mathcal{O}QT_{(12)})&=&\tr(\mathcal{O}QT_{(34)})=\tr(\mathcal{O}QT_{(1324)})=\tr(QT_{(1423)})\label{156}\\
\tr(\mathcal{O}QT_{(132)})&=&\tr(\mathcal{O}QT_{(234)})=\tr(\mathcal{O}QT_{(124)})=\tr(\mathcal{O}QT_{(143)})\\
\tr(\mathcal{O}QT_{(123)})&=&\tr(\mathcal{O}QT_{(142)})=\tr(\mathcal{O}QT_{(243)})=\tr(\mathcal{O}QT_{(134)})
\ea
\end{lemma}
\begin{proof} In order to prove the above relationship we need the expansion of $\mathcal{O}$ in Pauli operators
\be
\mathcal{O}=\sum_{\substack{P_1,P_2\in\mathcal{P}(2^N)\\P_3,P_4\in\mathcal{P}(2^N)}}=\tr(\mathcal{O}P_{1}\otimes P_2\otimes P_3\otimes P_4)P_{1}\otimes P_2\otimes P_3\otimes P_4
\ee
At this point we can prove the above relations for $\mathcal{O}\equiv P_{1}\otimes P_2\otimes P_3\otimes P_4$ without loss of generality. We won't perform all the calculations, rather just some instructive examples. Let us prove Eq.\eqref{154}
\be
\tr[(P_1\otimes P_2\otimes P_3\otimes P_4)Q]=\frac{1}{d^2}\sum_{P\in\mathcal{P}(2^N)}\tr(P_1P)\tr(P_2P)\tr(P_3P)\tr(P_4P)=d^2\sum_{P\in\mathcal{P}(2^N)}\delta_{PP_1}\delta_{PP_2}\delta_{PP_3}\delta_{PP_4}=d^2\delta_{P_1P_2P_3P_4}
\ee
the second one
\ba
\tr[(P_1\otimes P_2\otimes P_3\otimes P_4)QT_{(12)(34)}]&=&d^{-2}\sum_{P\in\mathcal{P}(2^N)}\tr(P_1PP_2P)\tr(P_3PP_4P)\\&=&\frac{1}{d^2}\sum_{P\in\mathcal{P}(2^N)} d^2\delta_{P_1P_2}\delta_{P_3P_4}K(P,P_4)K(P,P_2)\nonumber\\&=&\delta_{P_1P_2}\delta_{P_3P_4}\sum_{P\in\mathcal{P}(2^N)} \frac{1}{d}\tr((P_2P_4)^\dag P(P_3P_4)P)=\delta_{P_1P_2}\delta_{P_2P_4}\tr(P_2P_4)^2\nonumber\\&=&d^2\delta_{P_1P_2P_3P_4}\nonumber
\ea
the same procedure follows for $T_{(13)(24)}$ and $T_{(14)(23)}$. Let us prove Eq.\eqref{156}
\ba
\tr[(P_1\otimes P_2\otimes P_3\otimes P_4)QT_{(12)}]&=&\frac{1}{d^2}\sum_{P\in\mathcal{P}(2^N)}\tr(P_1PP_2P)\tr(P_3P)\tr(P_4P)=\delta_{P_3P_4}\tr(P_1P_3P_2P_3)\nonumber\\
&=&dK(P_2,P_3)\delta_{P_3P_4}\delta_{P_1P_2}
\ea
the $4$-cycle
\ba
\tr[(P_1\otimes P_2\otimes P_3\otimes P_4)QT_{(1423)}]&=&\frac{1}{d^2}\sum_{P\in\mathcal{P}(2^N)}\tr(P_1PP_3PP_2PP_4P)=\frac{1}{d^2}\sum_{P\in\mathcal{P}(2^N)}K(P,P_4)K(P,P_3)\tr(P_1P_3P_2P_4)\nonumber\\&=&\frac{1}{d^3}\sum_{P\in\mathcal{P}(2^N)}K(P,P_4P_3)\tr(P_1P_3P_2P_4)=d\delta_{P_{3}P_{4}}\delta_{P_1P_2}\tr(P_1P_3P_2P_4)\nonumber\\&=&dK(P_2,P_3)\delta_{P_1P_2}\delta_{P_3P_4}
\ea
the same calculations follow for $T_{(34)}$ and $T_{(1324)}$. \end{proof}

\subsection{$Q$ decomposition and traces\label{lemma}}
It is interesting to prove a useful property of the operator $Q$, that we recall is defined as
\be 
Q=\frac{1}{d^2}\sum_{P\in \mathcal{P}(2^N)} P^{\otimes 4}=\frac{1}{d^2}\sum_{\sigma_{i_{1}},\ldots,\sigma_{i_{N}}}(\sigma_{i_{1}}\otimes\cdots\otimes \sigma_{i_{N}})^{\otimes 4}
\ee 
where $\sigma_{i_{j}} \in \mathcal{P}(2)$.
It is possible to introduce a permutation $S\in \mathcal{S}_{4N}$, whose action on a tensor product state is defined as
\be 
S(\ket{i_{1}\ldots i_{N}}^{\otimes 4})=\ket{i_{1}}^{\otimes 4}\cdots\ket{i_{N}}^{\otimes 4}
\label{sdefinition}
\ee 
The adjoint action of a permutation $S$  on $Q$ reads 
\ba
SQS^{-1}&=&\frac{1}{d^2}\sum_{\sigma_{i_{1}},\ldots,\sigma_{i_{N}}}S(\sigma_{i_{1}}\otimes\cdots\otimes \sigma_{i_{N}})^{\otimes 4}S^{-1}\\
&=&\left(\frac{1}{4}\right)^{N}\sum_{\sigma_{i_{1}},\ldots,\sigma_{i_{N}}}\sigma_{i_{1}}^{\otimes 4}\otimes\cdots\otimes \sigma_{i_{N}}^{\otimes 4}=Q_{2}^{\otimes N}
\label{qtensors}
\ea
where $Q_{2}=\sum_{\sigma\in\mathcal{P}(2)}\sigma^{\otimes 4}$. Then, let us show the adjoint action of $S$ on a permutation operator between the $4$-copies of $\mathcal{H}$. Let $T_{\sigma}\in\mathcal{B}(\mathcal{H}^{\otimes 4})$ a permutation operator between $4$ copies of $\mathcal{H}\equiv \mathbb{C}^{2\otimes N}$ corresponding to $\sigma\in\mathcal{S}_4$; written in terms of bras and kets it reads
\be
T_{\sigma}=\sum_{\substack{i_{1}\ldots i_{N}\\
j_{1}\ldots j_{N}}}\sum_{\substack{
k_{1}\ldots k_{N}\\
l_{1}\ldots l_{N}}}\ket{\sigma(i_{1})\ldots\sigma(i_{N})\sigma(j_{1})\ldots\sigma(j_{N})\sigma(k_{1})\ldots\sigma(k_{N})\sigma(l_{1})\ldots\sigma(l_{N})}\bra{i_{1}\dots i_{N} j_{1}\dots j_{N} k_{1}\dots k_{N} l_{1}\dots l_{N}}
\ee
The adjoint action of $S$ reads
\be 
ST_{\sigma}S^{-1}=\sum_{\substack{i_{1}\ldots i_{N}\\
j_{1}\ldots j_{N}}}\sum_{\substack{
k_{1}\ldots k_{N}\\
l_{1}\ldots l_{N}}}\ket{\sigma(i_{1})\sigma(j_{1})\sigma(k_{1})\sigma(l_{1})}\bra{i_{1}j_{1}k_{1}l_{1}}\otimes\cdots\otimes \ket{\sigma(i_{N})\sigma(j_{N})\sigma(k_{N})\sigma(l_{N})}\bra{i_{N}j_{N}k_{N}l_{N}}
\ee 
It is clear that we can write
\be
ST_{\sigma}S^{-1}=T_{\sigma}^{(2)\otimes N}
\label{ttensors}
\ee
where we are denoting $T_{\sigma}^{(2)}=\sum_{i,j,k,l}\ket{\sigma(i)\sigma(j)\sigma(k)\sigma(l)}\bra{ijkl}\in\mathcal{B}(\mathbb{C}^2)$ which is a permutation operator between $4$ copies of a single qubit Hilbert space $\mathbb{C}^2$.

\begin{lemma}\label{lemma1}
Let $K_{i}$ and $K_{j}$ two identical single qubit gates with support on a different qubit, $i$ and $j$ respectively. Let $T_{\sigma}$ be a permutation operator between the $4$-copies of $\mathcal{H}$; the following equality holds
\be
\tr(T_{\sigma}QK_{i}^{\otimes 4}QK_{i}^{\dag\otimes 4})=\tr(T_{\sigma}QK_{j}^{\otimes 4}QK_{j}^{\dag\otimes 4})
\ee
\end{lemma}
\begin{proof} First of all
\be
K_{i}^{\otimes 4}=(\bbbone_{1}\otimes\dots\otimes \bbbone_{i-1}\otimes K\otimes \bbbone_{i+1}\otimes \dots \otimes \bbbone_{N})^{\otimes 4}
\ee
Acting adjointly with the permutation operator $S\in \mathcal{S}_{4N}$, defined in Eq.\eqref{sdefinition} on $K_{i}^{\otimes 4}$ we have
\be
SK_{i}^{\otimes 4}S^{-1}=(\bbbone_{1}^{\otimes 4}\otimes\dots\otimes \bbbone_{i-1}^{\otimes 4}\otimes K^{\otimes 4}\otimes \bbbone_{i+1}^{\otimes 4}\otimes \dots \otimes \bbbone_{N}^{\otimes 4})
\ee
Then from the above equality and from Eq.\eqref{qtensors} and Eq.\eqref{ttensors} we have
\ba
\tr(T_{\sigma}QK_{i}^{\otimes 4}QK_{i}^{\dag\otimes 4})&=&\tr(T_{\sigma}^{(2)\otimes N}Q_{2}^{\otimes N}(SK_{i}^{\otimes 4}S^{-1})Q_{2}^{\otimes N}(SK_{i}^{\dag\otimes 4}S^{-1}))\\&=&\tr(T_{\sigma}^{(2)}Q_{2}KQ_{2}K^{\dag})\tr(T_{\sigma}^{(2)\otimes (N-1)}Q_{2}^{\otimes (N-1)})\nonumber
\ea
where we used the fact that $(Q_{2}^{\otimes N-1})^2=Q_{2}^{\otimes N-1}$. From the above relation it's clear that the position of the qubit on which $K$ applies does not play any particular role. \end{proof}

\begin{application}\label{averagetrq}
Let $\psi = \otimes_i\psi_i$  a completely factorized random product state on $\mathbb{C}^{2\otimes N}$. Then
\be
\tr(\psi^{\otimes 4}Q)=d^{-1-(\log_25-2)}
\ee
\end{application}
\begin{proof}
Let us calculate $\tr(Q\psi^{\otimes 4})$ in the case $\psi=\ket{0}\bra{0}^{\otimes N}$. As proven in Sec. \ref{lemma}
\be
SQS^{-1}=(Q_{2})^{\otimes N}
\ee
where $S$ is a permutation operator defined in Eq.\eqref{sdefinition}. Here $Q_2$ reads
\be
Q_{2}=\frac{1}{4}(I+X+Y+Z)
\ee
therefore
\be
\tr(Q\psi^{\otimes 4})=\langle0|Q_2|0\rangle^{N}=\frac{1}{2^{N}}=d^{-1}
\ee
Now let us average $\tr(\psi^{\otimes 4}Q)$ with the local-qubit Haar average. Let $\psi_{loc}^{\otimes 4}$ be
\be
\psi_{loc}^{\otimes 4}=\int \prod_{i=1}^{N}\de U_{i}\left(\bigotimes_{i=1}^{N}U_{i}\psi_{i}U_{i}^{\dag}\right)^{\otimes 4}
\ee
where $\operatorname{supp}(U_{i})=\mathbb{C}^2$ for any $i$. Using the Haar average formulas displayed in Sec. \ref{haarintegration} the adjoint action of $S$ on $\psi_{loc}^{\otimes 4}$, defined in Eq.\eqref{sdefinition}, reads
\be
S\psi_{loc}^{\otimes 4 }S^{-1}=\left(\frac{\Pi_{\sym}^{(4)}}{D_{\sym}}\right)^{\otimes N}
\ee
where $\text{supp}(\Pi_{\sym}^{(4)})=\mathbb{C}^2$ and $D_{\sym}^{(4)}=(2\cdot 3\cdot 4\cdot 5)/24$. Therefore $\tr(\psi_{loc}^{\otimes 4}Q)$ is
\be
\tr(\psi^{\otimes 4} Q)=\tr\left(\left(\frac{\Pi_{\sym}^{(4)}}{D_{\sym}}\right)^{\otimes N}Q_{2}^{\otimes N}\right)=5^{-N}(\tr(\Pi_{\sym}^{(4)}Q_2))^{N}=\left(\frac{2}{5}\right)^{N}=d^{-1-(\log_25-2)}
\label{randomproductstate}
\ee
\end{proof}

\end{document}